\pdfoutput=1

\documentclass[a4paper]{article}

\usepackage[pages=all, color=black, position={current page.south}, placement=bottom, scale=1, opacity=1, vshift=5mm]{background}
\SetBgContents{

}      

\usepackage[margin=1in]{geometry} 

\usepackage{amsmath}
\usepackage{amsthm}
\usepackage{amssymb}

\usepackage[utf8]{inputenc}
\usepackage{hyperref}
\hypersetup{
	unicode,
	pdfauthor={Xi Chen, Cassandra Marcussen},
	pdftitle={Uniformity Testing over Hypergrids with Subcube Conditioning},
	pdfsubject={Uniformity Testing over Hypergrids with Subcube Conditioning},
	pdfkeywords={article, template, simple},
	pdfproducer={LaTeX},
	pdfcreator={pdflatex}
}

\theoremstyle{plain}
\newtheorem{theorem}{Theorem}

\newtheorem{lemma}[theorem]{Lemma}
\newtheorem{claim}[theorem]{Claim}

\newtheorem{proposition}[theorem]{Proposition}

\theoremstyle{definition}
\newtheorem{definition}[theorem]{Definition}

\newcommand{\reject}{\textsf{reject}}
\newcommand{\accept}{\textsf{accept}}

\usepackage{graphicx, color}
\usepackage{amsmath}
\usepackage{amssymb}
\usepackage{braket}
\usepackage{comment}
\usepackage{mathtools}
\usepackage{bbm}
\usepackage{setspace}
\usepackage{xcolor}
\graphicspath{{fig/}}

\usepackage{algorithm, algpseudocode, algorithmicx}
\usepackage{mathrsfs} 

\usepackage{lipsum}
\newcommand*{\email}[1]{
     \href{mailto:#1}{#1}\par
}

\title{Uniformity Testing over Hypergrids with Subcube Conditioning\vspace{0.2cm}}
\author{Xi Chen \\ Columbia University\\ \email{xichen@cs.columbia.edu}\and Cassandra Marcussen\\Harvard University\\\email{cmarcussen@g.harvard.edu}}

\usepackage{mathtools}

\DeclarePairedDelimiter\abs{\lvert}{\rvert}
\DeclarePairedDelimiter\norm{\lVert}{\rVert}

\makeatletter
\let\oldabs\abs
\def\abs{\@ifstar{\oldabs}{\oldabs*}}
\let\oldnorm\norm
\def\norm{\@ifstar{\oldnorm}{\oldnorm*}}
\makeatother

\begin{document}
\maketitle
\begin{abstract}
    We give an algorithm for testing   uniformity of distributions supported on hypergrids $[m_1]\times \cdots \times [m_n]$,
    which makes $\smash{\widetilde{O}(\text{poly}(m)\sqrt{n}/
    \epsilon^2)}$ many queries to a subcube conditional sampling oracle with $m=\max_i m_i$. When $m$ is a  constant, our algorithm is nearly optimal and  strengthens the algorithm of \cite{canonne2021random} which has the same query complexity but works for hypercubes $\{\pm 1\}^n$ only.

    A key technical contribution behind the analysis of our algorithm is a proof of a robust version of Pisier's inequality for functions over hypergrids 
    using Fourier analysis.
\end{abstract}
\newpage
\tableofcontents
\newpage

\section{Introduction}

Much of today’s data can be thought of as samples from an unknown probability distribution over a large and high-dimensional sample space. Testing global properties of a distribution \cite{CDVV14, VV14, Paninski08, DBLP:journals/corr/DiakonikolasK16, DBLP:journals/corr/DiakonikolasKN14} on such a space, however, is known to be intractable under the classical statistical model where an algorithm can only draw independent samples from the unknown distribution. This holds even for testing the property of \emph{uniformity}\footnote{Given an unknown distribution $p$ over a sample space $\Sigma$, accept with probability at least $2/3$ when $p$ is uniform over $\Sigma$ and reject with probability at least $2/3$ when $p$ is $\epsilon$-far from the uniform distribution in total variation distance.}, for which we know that $\Theta(\sqrt{|\Sigma|}/\epsilon^2)$ samples are both sufficient and necessary \cite{Paninski08, VV14}. 
So for high-dimensional sample spaces such as $\{\pm 1\}^n$ or $[{m_1}] \times \dots \times [{m_n}]$, the number of samples needed under the classical statistical model scales exponentially in $n$.
To circumvent the intractability, recent work has proceeded
by either restricting the class of input distributions (e.g., restricting $p$ to be a product distribution  \cite{DBLP:journals/corr/CanonneDKS16}), or by allowing stronger oracle access to $p$.
The goal of both approaches is to develop algorithms that scale polynomially or even sublinearly in the dimension $n$ under certain   well-motivated assumptions.

One of the most natural models in the latter direction is the \emph{subcube conditioning model}, which is particularly suitable for high-dimensional distributions \cite{CRS15, bc18, canonne2021random,CJLW20}. The model was suggested in \cite{CRS15} and first studied in \cite{bc18} (more discussion on the model and related work can be found in Section \ref{sec:relatedwork}). For the general space of $[{m_1}] \times \dots \times [{m_n}]$, subcube conditional query access allows algorithms to specify a subgrid  of $[m_1] \times \dots \times [m_n]$ by giving a restriction $\rho$ with $\rho_i\in [m_i]\cup \{*\}$ for each $i\in [n]$  and requesting a sample from the distribution conditioned\footnote{When conditioned on a subcube with zero support, one may consider models where the oracle returns either a uniform
sample or outputs ``error.'' We note that our algorithm will never run into this scenario.} on the sample lying in the subgrid specified by $\rho$ (i.e., the set of $x\in [{m_1}] \times \dots \times [{m_n}]$ with $x_i=\rho_i$ for all $i\in [n]$ such that $\rho_i\ne *$). 

Recently,  \cite{canonne2021random} gave an algorithm for testing uniformity over hypercubes $\{\pm 1\}^n$, which makes $\tilde{O}(\sqrt{n}/\epsilon^2)$ queries to a subcube conditional sampling oracle.
The algorithm is nearly optimal (given the $\Omega(\sqrt{n}/\epsilon^2)$ lower bound of \cite{DBLP:journals/corr/CanonneDKS16, DBLP:journals/corr/DaskalakisDK16} for testing uniformity of product distributions under the classical statistical model, which was observed in \cite{bc18} to carry over to subcube conditional sampling).
A drawback of their algorithm, however, is that it only works for hypercubes $\{\pm 1\}^n$.

\emph{In this paper, we study the problem of testing uniformity of distributions over the general hypergrid domain $[m_1] \times \dots \times [{m_n}]$ under the subcube conditioning model.} There are 
a number of compelling reasons  to study this problem.
From a practical perspective, testing algorithms for hypercubes are not applicable in scenarios when the variables\hspace{0.05cm}/\hspace{0.05cm}features are not Boolean.
And natural attempts to reduce the problem over $[{m_1}] \times \dots \times [{m_n}]$ directly to that over hypercubes do not seem to work either because the total variation distance is not preserved or the subcube conditional oracle does not cope with the reduction.
(For the latter consider the reduction from $[4]^n$ to $\{\pm 1\}^{2n}$ by encoding each entry of $[4]$ using two bits.
While the total variation distance is preserved, subcube conditional oracles for distributions over $\{\pm 1\}^{2n}$ cannot be simulated using those for $[4]^n$. The former corresponds to a more powerful oracle where an algorithm can, e.g., fix a coordinate to be in $\{1,3\}$.)

From a theoretical perspective, the problem is well-motivated due to a number of obstacles that one needs to be overcome to generalize the prior work of \cite{canonne2021random} from hypercubes to hypergrids. (More discussion and a comparison of our work with that of \cite{canonne2021random} can be found in Section \ref{sec:sketch}.) One of the primary  challenges is that their analysis of  correctness  crucially relies on a \emph{robust} version of Pisier's inequality. The latter is an inequality from convex analysis that relates the $\ell_s$-norm of a function $f$ over $\{\pm 1\}^n$ to
its $i$th coordinate Laplacian operator $L_i f$
  (see definition below), and it was not known whether a similar inequality holds for functions over hypergrids $[m_1] \times \dots \times [m_n]$.
  
\subsection{Our Contributions}

We study uniformity testing over hypergrids $[m_1] \times \dots \times [m_n]$ with $m_i \ge 2$ for all $i$. Let $m = \max_{i \in [n]} m_i$. Our main result is an algorithm that makes 
  $\smash{\tilde{O}(\text{poly}(m)\sqrt{n}/\epsilon^2)}$
  many subcube conditional queries:

\begin{theorem}[Uniformity Testing]\label{thm:subcube-alg-intro}
There is an algorithm which, given $n,m_1,\ldots,m_n$ and subcube conditional query access to a distribution $p$ supported on $[m_1] \times \dots \times [m_n]$ and a distance parameter $\epsilon \in (0,1)$, makes $\tilde{O}(m^{21} \sqrt{n}/\epsilon^2)$ queries and can distinguish with probability at least $2/3$ between the case when $p$ is uniform, and when $p$ is $\epsilon$-far from uniform in total variation distance.
\end{theorem}

Our algorithm improves the algorithm of \cite{canonne2021random} which only works for hypercubes, and is nearly optimal when $m$ is a constant. For general $m$, the best-known lower bound for the problem is $\Omega(\sqrt{nm}/\epsilon^2)$ \cite{bhattacharyya2021testing} (again, via the connection to uniformity testing of product distributions under the statistical model). While we believe that the polynomial $m^{21}$ in Theorem \ref{thm:subcube-alg-intro} can be improved by tightening up our analysis, it remains an important open question to pin down the complexity as a function of $n,m$ and $\epsilon$. 

We sketch the proof of Theorem \ref{thm:subcube-alg-intro} in Section \ref{sec:sketch}.
One of the main contributions of our paper is the proof of a robust version of Pisier's inequality for functions over hypergrids, which plays a crucial role in the analysis of the main algorithm and may be of independent interest. (We review and compare with the original Pisier's inequality \cite{PisierInequality} in Section \ref{pisier-proof-overview}.)
We need some notation to state the inequality. 
Since Fourier analysis will be used heavily in the proof of the inequality, from now on we will always use 
$$
\mathbb{Z}_M:=\mathbb{Z}_{m_1} \times \dots \times \mathbb{Z}_{m_n},$$ 
where $M=(m_1,\ldots,m_n)$, to denote the hypergrid. Given $x\in \mathbb{Z}_M$, $i\in [n]$ and $a\in \mathbb{Z}_{m_i}$, let $x^{(i) \to a}$ denote the vector obtained from $x$ by replacing $x_i$ with $a$. Given 
any function $f:\mathbb{Z}_M\rightarrow \mathbb{C}$,
  the \textit{$i$th~coordinate Laplacian operator}  is defined as
$$L_if(x) = f(x) - \mathbb{E}_{a\sim \mathbb{Z}_{m_i}} \big[f(x^{(i)\rightarrow a})\big].$$
For each $j\in [n]$, let $\mathbb{Z}_{m_j}^*=\{1,\ldots,m_j-1\}$, and 
  let $\omega_j = e^{2 \pi i /m_j}$ be the primitive $m_j$-th root of unity. 

We are ready to state the new Pisier's inequality for hypergrids.
The robust version is more involved; we state and discuss it later in Section \ref{pisier-proof-overview}.

\begin{theorem}[Pisier's Inequality for Hypergrids]\label{extended_pisier} Let $f: \mathbb{Z}_{M} \to \mathbb{C}$ be a function with $\mathbb{E}_{x \sim \mathbb{Z}_M}[f(x)] = 0$. 
Then, for any $s \in [1, \infty)$ we have 
\begin{align*}
\Big( \mathbb{E}_{x \sim \mathbb{Z}_M} \big[ \abs{f(x)}^s \big]\Big)^{1/s} \le O(\log n) \cdot \left(\mathbb{E}_{x, y \sim \mathbb{Z}_M} \left[  \abs{\sum_{i\in [n]}  L_i f(x) \sum_{a \in \mathbb{Z}_{m_i}^*} \omega_i^{-a y_i} \omega_i^{a x_i}  }^s \right]\right)^{1/s}.
\end{align*}
\end{theorem}

One should interpret Pisier's inequality as providing a way of connecting the $\ell_s$-norm of the function $f$~to~its Laplacian operators. Within the context of how the inequality is used in this paper,  Laplacian operators 
capture the difference in the function value along edges of the hypergrid (that is, between $x$ and $x^{(i) \to a}$). The extension of Pisier's inequality to hypergrids could have applications in other problems where the $\ell_s$-norm needs to be connected to edge-wise differences of a function defined over $\mathbb{Z}_M$.

\def\stars{\text{stars}}

\subsection{Proof Overview and Comparison with Previous Work}\label{sec:sketch}

First we recall the notion of restrictions and projections of a distribution.
Given a restriction $\rho$ with 
  $\rho_i$ $\in \mathbb{Z}_{m_i}\cup \{*\}$ for each $i\in [n]$, we write $\stars(\rho)$ to denote the set of $i\in [n]$ with $\rho_i=*$ and denote by $p_{|\rho}$ the distribution of ${x}_{\text{stars}(\rho)}$ with ${x}$ drawn from $p$ {conditioned} on $x_i=\rho_i$ for every $i \notin \text{stars}(\rho)$.
The other operation on distributions is  {projections}: Given  $S \subseteq [n]$,    $p_{S}$ denotes the distribution of $x_S$ with $x \sim p$. 

\begin{definition}[Random Restrictions] 
Given $\sigma \in [0,1]$, we let $\mathcal{S}_{\sigma}$ denote the distribution supported on subsets of $[n]$ where $S \sim \mathcal{S}_{\sigma}$ includes each   $i \in [n]$ independently with probability $\sigma$. Given a distribution $p$ supported on $\mathbb{Z}_M$, we use $\mathcal{D}_{\sigma}(p)$ to denote the following distribution of restrictions: to draw a restriction $\rho \sim \mathcal{D}_{\sigma}(p)$, we first sample a set $\smash{S \sim \mathcal{S}_{\sigma}}$ and an $x \sim p$; then, $\rho_i$ for each $i\in [n]$ is set to be
\begin{align} \label{eq:rr-intro}
\rho_i =
\begin{cases}
  * & \text{if } i \in {S} \\
  {x}_i & \text{if } i \notin {S}
\end{cases}.
\end{align}
\end{definition}

Given a distribution $p$ over $\mathbb{Z}_M$, we define the following bias vector that generalizes the mean vector of a distribution over $\{-1,1\}^n$:

\begin{definition} [Bias Vector] \label{mu-c_definition}
Let $p$ be a distribution over $\mathbb{Z}_M$, $i\in [n]$ and $c,d\in \mathbb{Z}_{m_i}$.
We define
$$\mu^{c,d}_i(p) = \frac{\Pr_{x \sim p}[x_{i} = c] - \Pr_{x \sim p}[x_{i} = d]}{\Pr_{x \sim p}[x_{i} = c] + \Pr_{x \sim p}[x_{i} = d]}.$$
When $\Pr_{x\sim p}[x_i=c]=\Pr_{x\sim p}[x_i=d]=0$, the bias $\mu_i^{c,d}(p)$ is set to be $0$ by default.
We also allow $c=d$ for notational convenience, in which case $\smash{\mu_i^{c,d}(p)=0}$ trivially.

Moreover, we write $\mu(p)$ to denote the \emph{bias vector} of $p$: $\mu(p)$ has $\sum_{i\in [n]} m_i^2$ entries $\mu_i^{c,d}(p)$ and thus,
$$
\big\|\mu(p)\big\|_2=\sqrt{
\sum_{i\in [n]}\sum_{c,d\in \mathbb{Z}_{m_i}}\left(\mu^{c,d}_i(p)\right)^2
}.
$$
\end{definition}

The intuition behind our uniformity testing algorithm is similar to  the algorithm of \cite{canonne2021random}, which is inspired by Lemma \ref{Lemma1.4} and Theorem \ref{Theorem1.5} below:

\def\calU{\mathcal{U}}

\begin{lemma} \label{Lemma1.4}
Let $p$ be a distribution over $\mathbb{Z}_M$. Then for any $\sigma \in [0, 1]$,  we have\footnote{We write $\calU$ to denote the uniform distribution, and $d_{TV}$ to denote the total variation distance.}
$$d_{TV}(p, \mathcal{U}) \leq \mathbb{E}_{S \sim \mathcal{S}_{\sigma}} \big[d_{TV}(p_{\overline{S}}, \mathcal{U}) \big] + \mathbb{E}_{\rho \sim \mathcal{D}_{\sigma}(p)} \big[  d_{TV}(p_{|\rho}, \mathcal{U}) \big].$$ 
\end{lemma}

\def\eps{\epsilon}

\begin{theorem} \label{Theorem1.5}
Let $p$ be a distribution over $\mathbb{Z}_M$. Then for any $\sigma \in [0, 1]$, we have
$$\mathbb{E}_{\rho \sim \mathcal{D}_{\sigma}(p)} \left[ \left\|\mu(p_{|\rho})\right\|_2 \right] \geq \frac{\sigma}{m^{7.5} \cdot\emph{\text{polylog}}(nm)} \cdot \tilde{\Omega}\left(\mathbb{E}_{S \sim \mathcal{S}_\sigma}\big[d_{TV}(p_{\overline{S}}, \mathcal{U}) \big] - 2 e^{- \min (\sigma, 1 - \sigma) n / 10} \right).$$
\end{theorem}

\def\calD{\mathcal{D}}

Lemma \ref{Lemma1.4} extends a 
  corresponding lemma from \cite{canonne2021random} for distributions supported on $\{-1, 1\}^n$ to distributions supported on hypergrids. 
The proof can be found in Appendix \ref{appendix-additionalproofs}.
Theorem \ref{Theorem1.5}, on the other hand, is the main technical result of the paper, which we discuss in the rest of the subsection.
But before that, assuming Lemma \ref{Lemma1.4} and Theorem \ref{Theorem1.5}, our main uniformity testing algorithm \textsc{SubCondUni} proceeds as follows (see Section \ref{sec:alg} for details):
consider a distribution $p$ over $\mathbb{Z}_M$ with $d_{TV}(p,\calU)\ge \eps$ and let $\sigma$ be sufficiently small; it will be set to be $1/\text{polylog}(1/\eps)$ in the proof. Then by Lemma \ref{Lemma1.4}, one of the following two cases must hold:
\begin{flushleft}
\begin{enumerate}
\item  $\mathbb{E}_{\rho \sim \mathcal{D}_{\sigma}(p)} [ d_{TV}(p_{|\rho}, \mathcal{U}) ] \geq \epsilon/2$: In this case, a typical draw of $\rho\sim \mathcal{D}_\rho(p)$ satisfies the property that $d_{TV}(p_{|\rho},\mathcal{U})$ remains large and the dimension $|\stars(\rho)|$ of $p_{|\rho}$ is much smaller than $n$ (i.e., $\approx \sigma n$).
This case is handled using recursive calls to $\textsc{SubCondUni}$ on $p_{|\rho}$ with $\rho\sim \calD_\rho(p)$.
\item $\mathbb{E}_{S \sim \mathcal{S}_\sigma}[d_{TV}(p_{\overline{S}}, \mathcal{U})] \geq \epsilon/2$:   Theorem \ref{Theorem1.5} implies that a typical $\rho\sim \calD_\rho(p)$ has a bias vector with a large $\ell_2$-norm.
This case is handled by  \textsc{ProjectedTestMean},
a subroutine we give in Section \ref{Mean-test-section} to decide whether a given  distribution is uniform or has a bias vector with a large $\ell_2$-norm.
\end{enumerate}\end{flushleft}

The performance guarantee of \textsc{ProjectedTestMean} is summarized in the following theorem:

\def\eps{\epsilon}

\begin{theorem} \label{projectedtestmean-performancelemma} 
There is an algorithm \textsc{ProjectedTestMean} which, given $n$, $M=(m_1,\ldots,m_n)$, $\epsilon>0$ and sample access to a  probability distribution $p$ over $\mathbb{Z}_M$, draws
$$
O\big(m^{4}\log (mn)\big)\cdot 
\max\left\{\frac{1}{\epsilon^2\sqrt{n}},\frac{1}{\epsilon}\right\}.
$$
many samples $x\sim p$ and 
satisfies the following properties:
\begin{enumerate}
    \item If $p$ is the uniform distribution, the algorithm outputs $\accept$ with probability at least $2/3$; and
    \item If $p$ satisfies $\|\mu(p)\|_2 \geq \epsilon \sqrt{n}$, the algorithm outputs $\reject$ with probability at least $2/3$.
\end{enumerate}
\end{theorem}

\textsc{ProjectedTestMean}
  generalizes the  \textsc{MeanTester} algorithm of \cite{canonne2021random}
  to work on distributions over hypergrids instead of just hypercubes. 
While it essentially reduces the task to the same task over hypercubes and passes it down to \textsc{MeanTester}, a few new ingredients are needed for the reduction to work. These include a preprocessing step (using a so-called \textsc{CoarseTest}) and a method to project the input distribution over $\mathbb{Z}_M$ to a small number of distributions over the hypercube, on which we run \textsc{MeanTester}.
We present  \textsc{ProjectedTestMean} and its analysis in Section \ref{Mean-test-section}.

We discuss the proof of Theorem \ref{Theorem1.5} in the rest of the overview.

\subsubsection{A Robust Pisier's Inequality for Hypergrids} \label{pisier-proof-overview}

The most important ingredient we need in the proof of Theorem \ref{Theorem1.5} is  a \emph{robust} version of Pisier's inequality for hypergrids. 
Pisier's inequality was first introduced in the paper \cite{PisierInequality}, and is an important result within the realm of convex analysis:
\begin{theorem}[Pisier's inequality \cite{PisierInequality}]\label{theo:original}
Let $f: \{\pm 1\}^n \to \mathbb{R}$ be a function with $\mathbb{E}_x[f(x)] = 0$. Then,
$$\mathbb{E}_{x \sim \{\pm 1\}^n } \big[ \abs{f(x)} \big] \le  
O(\log n) \cdot \mathbb{E}_{x \sim \{\pm 1\}^n } \left[ \abs{\sum_{i\in [n]} y_i x_i L_i f(x)} \right].$$
\end{theorem}

Alternative proofs of Pisier's inequality can be found in \cite{NS02, BV87}.
Before this work we are~not aware of any generalization of Pisier's inequality beyond
  the domain of hypercubes. 
The flow of our~proof of Theorem \ref{projectedtestmean-performancelemma}, Pisier's inequality over hypergrids, at a high level follows that of \cite{NS02}. 
The challenge lies in careful considerations  
  required when moving from the use of Fourier analysis on hypercubes, where functions in the basis are $\{\pm 1\}$-valued, to 
  Fourier analysis on hypergrids, where functions in the basis take complex values that are roots of unity.
In particular, more intricate expressions needed to be discovered  (e.g., in Lemma \ref{pisier_subpart})
  in order to obtain cancellations that help connect the $\ell_s$-norm of $f$ with
  its Laplacian operators.

As mentioned earlier, Theorem \ref{extended_pisier} is not sufficient for our purpose of proving Theorem \ref{Theorem1.5} but we need a more powerful, robust version of Pisier's inequality for hypergrids. Additional definitions are needed~to state the inequality so we delay it 
  to Section \ref{section:robust-pisier} (Theorem \ref{extended_robust_pisier}) where it is proved.
The notion of robustness is the same as in \cite{KMS18} and \cite{canonne2021random}:
in the original Pisier's inequality such as Theorem \ref{theo:original} and \ref{extended_pisier}, 
  difference in the function value along each edge $\{x,x^{(i)\to a}\}$ 
  is accounted twice on the RHS, once at~$x$ and once at $x^{(i)\to a}$; in the robust version, imagine that an adversary gets to pick any orientation of edges and the same inequality still needs to hold when each directed
  edge $(x,x^{(i)\to a})$
  (oriented by the adversary from $x$ to $x^{(i)\to a}$) is only accounted once at $x$ (but not at $x^{(i)\to a}$).
This robustness will be crucial when we apply the inequality to prove Theorem \ref{Theorem1.5}, which we discuss in Section \ref{sec:sketchlemma}.

The key observation behind the proof of our robust Pisier's inequality for hypergrids is similar
  to~that of \cite{canonne2021random} for hypercubes: at one point of the proof of the original inequality, 
  the expectation where every edge of the hypergrid is accounted twice (once for each of its 
  vertices) can be replaced by a similar expectation where every edge is accounted once 
  with respect to a given orientation of edges.
Again,~the need to work with Fourier analysis over $\mathbb{Z}_M$ and deal with roots of unity makes the analysis much more demanding.
Indeed the inequality we prove takes a more complex form on the RHS compared to 
  Theorem \ref{extended_pisier}; in contrast, the robust Pisier's inequality for hypercubes of 
  \cite{canonne2021random} looks identical to Theorem \ref{theo:original}.

\subsubsection{Proof of Theorem \ref{Theorem1.5}}
\label{sec:sketchlemma}

We start with some notation.
Let $\mathcal{S}(t)$ be the uniform distribution supported on all subsets of $[n]$ of size $t$. Given a distribution $p$ over $\mathbb{Z}_M$, we let $\mathcal{D}(t, p)$ be the following distribution over restrictions: to draw  $\rho\sim \mathcal{D}(t,p)$, 
  we first sample $y\sim p$ and $S\sim \mathcal{S}(t)$ and then set $\rho$ to be $\rho_i=*$ if $i\in S$ and $\rho_i=y_i$ if $i\notin S$.

We are now ready to state the main technical lemma, which is proved in Section \ref{prooflemma9}:

\begin{lemma} \label{Lemma3.1}
Let $p$ be a distribution over $\mathbb{Z}_M$, $t \in [n-1]$, and denote
$$\alpha := \mathbb{E}_{T \sim \mathcal{S}(t)} \big[ d_{TV}(p_{\overline{T}}, \mathcal{U}) \big] \geq 0.$$
Then we have
\begin{equation} \label{equation17}
\begin{aligned}
    &\mathbb{E}_{\rho \sim \mathcal{D}(t, p)} \Big[ 
    \left\|\mu(p_{|\rho})\right\|_2
    \Big] + 
    \mathbb{E}_{\rho  \sim \mathcal{D}(t + 1, p)} \Big[ 
    \left\|\mu(p_{|\rho})\right\|_2
    \Big] 
    \ge \frac{t}{n} \cdot \frac{\alpha}{m^{7.5}\cdot
    \mathrm{polylog}(nm/\alpha)}.
\end{aligned}
\end{equation}
\end{lemma}

\def\calS{\mathcal{S}}

A significant portion of the paper is dedicated to proving Lemma \ref{Lemma3.1} in Section \ref{prooflemma9}. Once Lemma \ref{Lemma3.1} is proven, it only requires a short proof to obtain Theorem \ref{Theorem1.5}.
The proof of Theorem \ref{Theorem1.5} assuming Lemma \ref{Lemma3.1} is very similar to an argument used in  \cite{canonne2021random}, except for a minor change. It is included in Appendix \ref{appendix-additionalproofs} for completeness.
Below we sketch the proof of Lemma \ref{Lemma3.1} and compare it with \cite{canonne2021random}.

\def\calH{\mathcal{H}}

The key step of the proof of Lemma \ref{Lemma3.1} is the construction of a family of 
  directed graphs that is used to connect
  $d_{TV}(p_{\overline{T}},\calU)$ of $T\sim \calS(t)$ with $\|\mu(p_{|\rho})\|_2$ of
  either $\rho\sim \calD(t,p)$ or $\sim \calD(t+1,p)$.
In more details, let $T\subset [n]$ be a set of size $t$ and let $K=M_{\overline{T}}$ (so 
  $p_{\overline{T}}$ is a distribution over the hypergrid $\mathbb{Z}_K$).
Let $\calH(T)$ denote the undirected graph over $\mathbb{Z}_K$ with undirected edges $\{x,x^{(i)\to b}\}$ for all $x\in \mathbb{Z}_K,i\in \overline{T}$ and $b\ne x_i$.
Using values of $p_{\overline{T}}(x)$, we classify edges of $\calH(T)$ into those that are
  \emph{uneven} and \emph{even}.
Roughly speaking, an edge $\{x,y\}$ of $\calH(T)$ is uneven if 
$$\max\big(p_{\overline{T}}(x),p_{\overline{T}}(y)\big)
\gg \min\big(p_{\overline{T}}(x),p_{\overline{T}}(y)\big)$$
and is even otherwise.
The most important step of the proof is the construction of an orientation $G(T)$ of $\calH(T)$,
  where different strategies are used to orient uneven edges and even edges.

Once the directed graphs $G(T)$ for each $T$ is in place,
  the proof of Lemma \ref{Lemma3.1} proceeds as follows:
\begin{flushleft}\begin{enumerate}
\item In Section \ref{tv_to_directed_graphs}, we apply the robust Pisier's inequality
  on $s=1$ and $f$ over $\mathbb{Z}_K$ set to be 
  $$
  f(x)=p_{\overline{T}}(x)\prod_{j\in \overline{T}} m_j -1
  $$
  so the expectation of $f$ is $0$ and the LHS of the inequality is exactly $d_{TV}(p_{\overline{T}},\calU)$. Orienting the edges using $G(T)$, the inequality implies that either the expectation of
\begin{equation}\label{eq:case1}
\sqrt{\text{number of outgoing uneven edges of $x$}}
\end{equation} or the expectation of
\begin{equation}\label{eq:case2}
\sqrt{\text{number of outgoing even edges of $x$}}
\end{equation}
   when $x\sim p_{\overline{T}}$ is large in terms of $d_{TV}(p_{\overline{T}},\calU)$.
   
\item On the other hand, in Section \ref{sec:haha1}, \ref{sec:haha2} and 
  \ref{Case_2_section}, we connect $G(T)$ with $\|\mu(p_{|\rho})\|_2$ by showing that 
  when either the expectation of (\ref{eq:case1}) or the expectation of (\ref{eq:case2})
  is large for a typical $T\sim \calS(t)$,
  it implies that $\|\mu(p_{|\rho})\|_2$ is large for a typical $\rho$ drawn either from
  $\calD(t,p)$ or $\calD(t+1,p)$. The two cases of (\ref{eq:case1}) and (\ref{eq:case2}) are handled separately in Section \ref{sec:haha2} and 
  \ref{Case_2_section}, respectively.
\end{enumerate}\end{flushleft}

Compared to that of a similar lemma in \cite{canonne2021random}, our proof of Lemma \ref{Lemma3.1} differs significantly in~the construction of directed graphs $G(T)$
  due to the simple fact a vertex in $\calH(T)$ has multiple edges along the 
  same variable (while in hypercubes, every vertex has a unique edge along each variable).
In particular, the orientation of even edges needs to be handled 
  with a more delicate strategy. In \cite{canonne2021random}, uneven and even edges are oriented separately;
  in contrast, the orientation of even edges here crucially
  depends on that of uneven edges (uneven edges are handled first, followed by even edges).
The analysis in Case 2 (Section \ref{Case_2_section}), which becomes more involved compared to \cite{canonne2021random},  only works with the 
  new orientation strategy for even edges.

\subsection{Background and Related Work}\label{sec:relatedwork}

\textbf{Distribution testing:} Distribution testing --- initially studied in \cite{GR00}, \cite{BFRSW00}, and \cite{BFRSW13} --- is concerned with determining whether a  distribution satisfies a certain property or is \textit{far} from satisfying the property. Sample-optimal algorithms are known for a range of problems in distribution testing in the standard setting, where samples are drawn independently from the probability distribution that is being tested. For example, \cite{CDVV14,VV14,Paninski08,DBLP:journals/corr/DiakonikolasK16,DBLP:journals/corr/DiakonikolasKN14} give algorithms for testing with optimal sample complexity. The sample complexity lower bounds for many such problems have a polynomial dependence on the domain size, which in the high-dimensional setting leads to an exponential dependence on the dimension. Therefore, in this setting, newer models of sampling or testing are needed to achieve an improved (ideally sublinear) dependence on the dimension. \medskip

\noindent\textbf{Property testing on extended high-dimensional domains:} Property testing on extended high-dimensional domains $[m]^n$, also known as hypergrids, is fruitful to study due to its potential to yield sample complexity bounds that depend explicitly on both the alphabet size $m$ and the dimension $n$. The goal of this research is typically to construct algorithms with a polynomial (or even sublinear) dependence on $n$ and a polynomial dependence on $m$. Many of the algorithms or sample complexity lower-bounds for testing properties of functions or distributions over hypergrids rely on Fourier analysis. For example, in \cite{BRY14}, Blais, Raskhodnikova, and Yaroslavtsev utilize a set of Walsh functions, a canonical Fourier basis for functions on the line $[m]$, in their analysis. Other papers such as \cite{black2017od} and \cite{downsampling_walsh} also apply Fourier analysis by using Walsh functions over $[m]^n$. One advantage of Walsh functions is that they are $\{\pm 1\}$-valued. 
However, Walsh functions can only be used as a Fourier basis if $m$ is a power of $2$. In this paper, we consider the high-dimensional domain $\mathbb{Z}_{m_1} \times \dots \times \mathbb{Z}_{m_n}$ and use certain powers of the primitive $m$-th root of unity to form the Fourier basis, which allows us to avoid this restriction to powers of $2$.

Other papers on testing properties over hypergrids $[m]^n$ include \cite{CS13,CS12,NIPS2015_1f36c15d,AJMR12}.\medskip

\noindent\textbf{Subcube conditioning:}  As mentioned earlier, under the standard sampling model, $\sqrt{|\Sigma|} / \epsilon^2$ samples are~needed for testing uniformity \cite{Paninski08,VV14}. To circumvent this issue under the high-dimensional setting, one may choose to consider distributions with more structure; for example, product distributions \cite{DBLP:journals/corr/CanonneDKS16}. On the other hand, one may study query models with stronger access to the  distribution. The subcube conditioning oracle model was studied with the latter purpose, and was first introduced in \cite{CRS15} and studied in \cite{bc18}. 

In \cite{canonne2021random}, the authors give a nearly-optimal uniformity testing algorithm for distributions supported on hypercubes $\{\pm 1\}^n$ under the subcube conditioning model. The subcube conditioning oracle model has also been used in studying the problems of learning and testing junta distributions on $\{\pm 1\}^n$ with respect to the uniform distribution \cite{CJLW20}. 

The subcube conditioning model is a theoretical model, but not an artificial model. Subcube conditional samples provide stronger access to the underlying distribution that is also potentially practically realizable. Subcube conditioning has received recent attention beyond the field of property testing, for example in \cite{blanc2023lifting}. In \cite{blanc2023lifting}, the authors study how subcube conditioning can be used to convert PAC learning algorithms that work under the uniform distribution into ones that works under an arbitrary and unknown distribution. 

Other recent papers have used variations of the subcube conditioning model to study the problem of identity testing. In \cite{BCSV}, the authors prove that if approximate tensorization holds for the visible distribution $\mu$ over $[ k]^n$, then there is an efficient identity testing algorithm for any hidden distribution $\pi$ using $\tilde{O}(n/\epsilon)$ queries to the so-called \textit{coordinate} {oracle}. The latter is similar to the subcube oracle, with the added restriction that all but one coordinate must be fixed when taking a random restriction.

More broadly, the subcube conditioning model is an adaptation of the conditional sampling model. The original and more general conditional sampling model \cite{CFGM13, CFGM16, CRS14,CRS15} allows for the algorithm to specify an arbitrary subset of a domain and receive a sample that is conditioned on it lying in the subset.  
This conditional sampling model has been applied to a range of problems in distribution testing and beyond in order to circumvent lower bounds in the 
standard sampling model.

\subsection{Notation} \label{section:notation}

We use $\tilde{O}(f(n))$ to denote $O(f(n)\cdot \text{polylog}(f(n)))$ and $\tilde{\Omega}(f(n))$ to denote $\Omega(f(n)/(1 + |\text{polylog}(f(n))|))$.
We write $f(n) \lesssim g(n)$ if, for some constant $c > 0$, $f(n) \leq c \cdot g(n)$ for all $n \geq 0$. $\gtrsim$ is defined similarly. 

Given $M=(m_1,\ldots,m_n)$,
 we write $\mathbb{Z}_M$ to denote $\mathbb{Z}_{m_1} \times \dots \times \mathbb{Z}_{m_n}$. We will   occasionally denote   $${\bigtimes_{i \in [n]} [m_i]}:={[m_1] \times \dots \times [m_n]}\quad \text{\ and\ }\quad \bigtimes_{i \in [n]} \mathbb{Z}_{m_i}:={\mathbb{Z}_{m_1} \times \dots \times \mathbb{Z}_{m_n}}.
 $$
Given $x \in \{-1, 1\}^n$, we write $x^{(i)}$ to denote the string that is identical to $x$ but with coordinate $i$ flipped, i.e. $\smash{x_j^{(i)} = x_j}$ for all $j \neq i$ and $\smash{x_i^{(i)} = - x_i}$.
Given $x\in \mathbb{Z}_M$,  $i\in [n]$ and $a\in \mathbb{Z}_{m_i}$, we let $x^{(i) \to a}$ denote the vector that is identical to $x$ but with coordinate $i$ set to $a$, i.e. $\smash{x^{(i) \to a}_j = x_j}$ for all $j \neq i$ and $\smash{x^{(i) \to a}_i} = a$.

\section{The Algorithm}\label{sec:alg}

In this section we present our main testing algorithm, \textsc{SubCondUni}, 
and use it to prove Theorem \ref{thm:subcube-alg-intro}.
It is presented as Algorithm \ref{algo:subconduni} and uses \textsc{ProjectedTestMean} as a subroutine.

\begin{algorithm}[t]
\algblock[Name]{StartMainCase}{EndMainCase}
	\begin{algorithmic}[1]
		\Require Dimension $n$, $M=(m_1,\ldots,m_n)$, subcube access to distribution $p$ over $\mathbb{Z}_M$, and $\epsilon \in (0,1/2]$ 
		\If{$n$ and $\epsilon$ violate (\ref{equation11})} 
		     \State Run an algorithm (Lemma 4.20)
       from  \cite{BCSV})
		        and \Return the same answer
		\EndIf
		\StartMainCase\Comment{Main case: (\ref{equation11}) satisfied}
		  \State Let $m = \max_{i} m_i$ and let $\smash{L=L(n,\epsilon)}$ be as defined in (\ref{equation15}) 

		\For {$j=1,2,\dots ,\lceil\log 2L\rceil$} \label{scalgo:first-for-start}
		\Comment{Handle the first case of  (\ref{equation13})}

			\State Sample 
			$s_j=8L \log (2L)\cdot 2^{-j}$ restrictions from $\mathcal{D}_{\sigma}(p)$ \label{scalgo:sample-restriction}			
			\For {each restriction $\rho$ sampled with $|\text{stars}(\rho)|>0$}
			    \State Run $\textsc{ProjectedTestMean}(|\text{stars}(\rho)|, M_{\stars(\rho)}, 2^{-j}, p_{| \rho})$ for $r=O(\log (nm/\epsilon))$ times 
       \label{scalgo:run-ptmean} 
			    \State \Return $\reject$ if the majority of calls return $\reject$ \label{scalgo:reject-ptmean-majority}
			   
		    \EndFor
		 
		\EndFor \label{scalgo:first-for-end}
		\For {$j=1,2,\dots,\lceil \log (4/\epsilon)\rceil$} \Comment{Handle the second case of (\ref{equation13})} 
			      \State Sample 
			      $s_j' = (32/\epsilon)\log (4/\epsilon)\cdot 2^{-j} $  restrictions from $\mathcal{D}_{\sigma}(p)$
			      \For {each restriction $\rho$ sampled that satisfies $0<|\text{stars}(\rho)|\leq 2\sigma n$} 
			     \State Let $M_{\text{stars}(\rho)} = (m_{i_1}, m_{i_2}, \dots, m_{i_{|\text{stars}(\rho)|}})$ where each $i_k \in \text{stars}(\rho)$       
				\State Run $\textsc{SubCondUni}(|\text{stars}(\rho)|, M_{\text{stars}(\rho)}, 2^{-j},p_{|\rho})$ for
					$t=100 \log (16/\epsilon)$ times
					\State \Return $\reject$ if the majority of calls return \reject \label{scalgo:recursive-call-reject}
			      \EndFor
			 \EndFor
			\State \Return $\accept$
		\EndMainCase
		
  	\end{algorithmic}
	\caption{$\textsc{SubCondUni}( n, M, \epsilon,p )$}\label{algo:subconduni}
\end{algorithm}

\def\SubC{\textsc{SubCondUni}}

This algorithm proceeds in a similar fashion as the \textsc{SubCondUni} algorithm in \cite{canonne2021random}. 
Let $p$ be a distribution over $\mathbb{Z}_M$ with $M=(m_1,\ldots,m_n)$ (where each $m_i \geq 2$).
Let
\begin{equation}\label{eq:defc0}
\sigma := \sigma(\epsilon) = \frac{1}{C_0 \cdot \log^4(16/\epsilon)}
\end{equation}
where $C_0 > 0$ is a sufficiently large constant to be specified later. 
If $n$ and $\eps$ violate
\begin{equation} \label{equation11}
    e^{-\sigma n /10} \leq \epsilon/8,
\end{equation}
i.e., $\eps$ is tiny, we use a result from \cite{BCSV}. 
Given the violation of (\ref{equation11}),
  we have
\begin{equation} \label{equation12}
    n = O \left( \frac{1}{\sigma} \cdot \log \left( \frac{1}{\epsilon}\right) \right) = O \left( \log^5 \left( \frac{1}{\epsilon}\right) \right).
\end{equation}
It follows from Theorem 4.1 and Lemma 4.20 of \cite{BCSV} 
 that uniformity can be tested with $\widetilde{O}({n \sqrt{m}}/{\epsilon^2})$ samples in this case. 
(To use  Theorem 4.1 and Lemma 4.20 in \cite{BCSV}, note that the uniform distribution is $\nu = ({1}/{m})$-balanced and satisfies approximate tensorization of entropy with constant $C = 1$.)

From now on we focus on the general case where $n$ and $\eps$ satisfy (\ref{equation11}).
To better understand how the algorithm works, consider the case when $d_{TV}(p, \mathcal{U}) \geq \epsilon$. Lemma \ref{Lemma1.4} implies that either
\begin{equation}
\mathbb{E}_{S \sim \mathcal{S}_\sigma}\big[d_{TV}(p_{\overline{S}}, \mathcal{U})\big] \geq \epsilon/2\quad \text{\ or\ }\quad
    \mathbb{E}_{\rho \sim \mathcal{D}_{\sigma}(p)} \left[ d_{TV}(p_{|\rho}, \mathcal{U}) \right] \geq \epsilon/2. \label{equation13}
\end{equation}

Assuming $\mathbb{E}_{S \sim \mathcal{S}_\sigma}[d_{TV}(p_{\overline{S}}, \mathcal{U})] \geq \epsilon/2$, using $e^{-\sigma n /10} \leq \epsilon/8$ and $\sigma = 1/\text{polylog}(1/\epsilon)$, Theorem \ref{Theorem1.5} gives
\begin{equation} \label{equation14}
\mathbb{E}_{\rho \sim \mathcal{D}_{\sigma}(p)} \left[ \frac{1}{m\sqrt{n}} \left\|\mu(p_{|\rho})\right\|_2 \right] \geq \widetilde{\Omega} \left(\frac{\epsilon}{m^{8.5} \sqrt{n}} \right).
\end{equation}
This is handled in the first for-loop of the main case, where \textsc{ProjectedTestMean} is used as a subroutine to tell whether $p_{|\rho}$ is uniform or has a bias vector with a large $\ell_2$-norm.
(Note that subcube conditional query access to $p$ is used to simulate sample access to $p_{|\rho}$ needed by \textsc{ProjectedTestMean}.)
The parameter $L$ used in this for-loop is defined as the inverse of the RHS of (\ref{equation14}) so $L$ satisfies
\begin{equation} \label{equation15}
L := L(n, \epsilon) = \widetilde{O}\left(\frac{m^{8.5} \sqrt{n}}{\epsilon}\right).
\end{equation} 

For the other case when $\mathbb{E}_{\rho \sim \mathcal{D}_{\sigma}(p)} [ d_{TV}(p_{|\rho}, \mathcal{U}) ] \geq \epsilon/2$, note that a typical draw of $\rho\sim \mathcal{D}_\rho(p)$ satisfies both that $d_{TV}(p_{|\rho},\mathcal{U})$ remains large and that the dimension $|\stars(\rho)|$ of $p_{|\rho}$ is much smaller (i.e., $\approx \sigma n$).
Intuitively this case is handled using recursive calls to $\textsc{SubCondUni}$
  in the second for-loop.
(Note that subcube conditional query access to $p$ can be used to simulate subcube conditional query access to $p_{|\rho}$ needed by recursive calls to \textsc{SubCondUni}.)

The proof of Theorem \ref{thm:subcube-alg-intro} using $\textsc{SubCondUni}$ follows similar arguments used in the proof of Theorem 2.1 of \cite{canonne2021random}.
We included the proof in Appendix  \ref{appendix-additionalproofs} for completeness.

\section{Robust Pisier's Inequality on Hypergrids} \label{section:robust-pisier}

\def\calH{\mathcal{H}}

In this section, we prove a robust version of Pisier's inequality on functions over hypergrids $\mathbb{Z}_{m_1} \times \dots \times \mathbb{Z}_{m_n}$, for $m_1, \dots, m_n \geq 2$. This inequality will play a crucial role in the proof of the main technical lemma, Lemma \ref{Lemma3.1}. First, let $\mathbb{Z}_{M}$ denote $\mathbb{Z}_{m_1} \times \dots \times \mathbb{Z}_{m_n}$, where $M = (m_1, \dots, m_n)$. 

Our notion of \textit{robustness} is similar to that of \cite{canonne2021random} and \cite{KMS18}, where the inequality holds for any orientation of an undirected graph $\calH$ over $\mathbb{Z}_M$, which will be defined below.
  
Given $x\in \mathbb{Z}_{M}$, $i\in [n]$ and $a\in \mathbb{Z}_{m_i}$, we write $x^{(i) \to a}$ to denote the vector that satisfies $x^{(i) \to a}_j = x_j$ for all $j \neq i$, and $\smash{x^{(i) \to a}_i} = a$.
Let $\mathcal{H}$ be the undirected graph over $\mathbb{Z}_M$ that consists of undirected edges $\{x, x^{(i) \to a} \}$ for all $x \in \mathbb{Z}_M$, $i \in [n]$, and $a \in \mathbb{Z}_{m_i}$ such that $x_i\ne a$.
(Equivalently, $\{x,y\}$ is an undirected edge in $\calH$ if there exists an $i\in [n]$ such that $x_i\ne y_i$ and $x_j=y_j$ for all $j\ne i$.)

Consider a function $f:\mathbb{Z}_M\rightarrow \mathbb{C}$, we recall the definition of the 
  Laplacian operator $L_if$ (\cite{odonnell2021analysis}):

\begin{definition}
Let $i \in [n]$ and $f: \mathbb{Z}_M \to \mathbb{C}$. The \textit{ith coordinate Laplacian operator} $L_i f$ is defined by:
$$L_if(x) = f(x) - \mathbb{E}_{a\sim \mathbb{Z}_{m_i}} \big[f(x^{(i)\rightarrow a})\big].$$
Given $a\in \mathbb{Z}_{m_i}$, we define $L_i^a f$
  to be 
$$L_i^a f(x) = \frac{f(x) - f(x^{(i) \to a})}{m_i}.$$
So we have $L_if(x)=\sum_{a\in \mathbb{Z}_{m_i}} L_i^af(x)$.
\end{definition}

For each $j\in [n]$, let
$\mathbb{Z}_{m_j}^*=\{1,\ldots,m_j-1\}$ and 
  let $\omega_j = e^{2 \pi i /m_j}$ be the primitive $m_j$-th root of unity.
We are now ready to state our robust Pisier's inequality
  for functions over $\mathbb{Z}_M$:

\begin{theorem}[Robust Pisier's Inequality for Functions over $\mathbb{Z}_M$]\label{extended_robust_pisier} Let $f: \mathbb{Z}_M \to \mathbb{C}$ be a function with $\mathbb{E}_{x \sim \mathbb{Z}_M}[f(x)] = 0$ and let $G$ be an orientation of $\mathcal{H}$. 
Then for any $s \in [1, \infty)$ we have
\begin{align*}
&\Big( \mathbb{E}_{x \sim \mathbb{Z}_M} \big[ \abs{f(x)}^s \big]\Big)^{1/s}\\ &\hspace{0.5cm}\lesssim \log (n) \cdot \left(\mathbb{E}_{x, y \sim \mathbb{Z}_M} \left[  \abs{\sum_{i\in [n]}  \sum_{a \in \mathbb{Z}_{m_i}^*} \sum_{\substack{d \in \mathbb{Z}_{m_i}^* \\ (x, x^{(i) \to x_i + d}) \in G}}  \hspace{-0.4cm} (1-\omega_i^{ad})\omega_i^{-a y_i} \omega_i^{a x_i} L_i^{x_i + d} f(x)  }^s \right]\right)^{1/s}.
\end{align*}
\end{theorem}

\subsection{Fourier Analysis over Hypergrids}
Given $M=(m_1,\ldots,m_n)$, we will use Fourier analysis over $\mathbb{Z}_{M}$ (see \cite{CSPFourier,odonnell2021analysis}). To this end,
we represent any function $f: \mathbb{Z}_M \to \mathbb{C}$ using this Fourier basis with the following expression:
$$f(x) = \sum_{u \in \mathbb{Z}_M} \hat{f}(u) \cdot \prod_{i\in [n]}\omega_i^{u_i x_i},$$
where the Fourier coefficients $\hat{f}(u)$ are given by
$$
\hat{f}(u)=\left(\prod_{i = 1}^n\frac{1}{m_i}\right)\sum_{x\in\mathbb{Z}_M}
f(x)\cdot
{\prod_{i\in [n]}\omega_i^{-u_ix_i}} .
$$
We will use the following fact about the $i$th coordinate Laplacian operator: \begin{equation}
\label{hehe5}L_if(x) = \sum_{u \in \mathbb{Z}_M: u_i \neq 0} \hat{f}(u)\cdot \prod_{i\in [n]}\omega_i^{u_i  x_i}.
\end{equation}

Given $\rho\in [0,1]$ and $x \in \mathbb{Z}_M$, we write $N_{\rho}(x)$ to denote the following distribution supported on $\mathbb{Z}_M$: To sample $y \sim N_{\rho}(x)$, for each $i \in [n]$ we set $y_i = x_i$ with probability $\rho$, and set $y_i$ to be a uniform random number from $\mathbb{Z}_{m_i}$ with probability $1 - \rho$.

\begin{definition}[Noise Operator] Given $f:\mathbb{Z}_M\rightarrow \mathbb{C}$ and  $\rho \in [0, 1]$, the noise operator   $T_\rho$ is defined as
$$T_\rho f(x) = \mathbb{E}_{y \sim N_\rho (x)} \big[f(y)\big].$$
\end{definition} 

Given $u \in \mathbb{Z}_M$, we write 
$$\text{supp}(u) = \{i\in [n] : u_i \neq 0 \}\quad\text{and}\quad\# u = | \text{supp}(u)|.$$
The following proposition relates the noise operator to its Fourier expansion. The proposition can be found in
  \cite{odonnell2021analysis} for the case of $m_1 = \dots = m_n$.

\begin{proposition}\label{hehe1}
Let $\rho \in [0, 1]$ and let $f: \mathbb{Z}_M \to \mathbb{C}$. Then, the Fourier expansion of $T_\rho f$ is given by:
$$T_\rho f(x) = \sum_{u \in \mathbb{Z}_M} \rho^{(\# u)} \hat{f}(u)\cdot  \prod_{i\in [n]}\omega_i^{u_i x_i}.$$
\end{proposition}

\begin{proof}
By the definition of $T_\rho$ we have
 $$T_\rho f(x) = \mathbb{E}_{y \sim N_\rho(x)}\big[f(y)\big] = \sum_{u \in \mathbb{Z}_M} \hat{f}(u) \cdot \mathbb{E}_{y \sim N_\rho(x)}\left[\prod_{i\in [n]}\omega_i^{u_iy_i}\right].$$
 Next, we have $$\mathbb{E}_{y \sim N_\rho(x)}\left[\prod_{i\in [n]}\omega_i^{u_iy_i}\right] = \prod_{i \in[n]} \left(\rho \cdot \omega_i^{u_i x_i} + \frac{1 - \rho}{m_i} \sum_{z \in \mathbb{Z}_{m_i}} \omega_i^{u_i z}\right) = \rho^{(\# u)}\prod_{i\in [n]}\omega_i^{u_ix_i}.$$
 This is because if $u_i \neq 0$, then the sum of  $\omega_i^{u_i z}$ is $0$ and the $i$-th coordinate contributes $\rho \cdot \omega_i^{u_i x_i}$ to the product. If $u_i = 0$,  the $i$-th coordinate contributes $\rho + 1 - \rho = 1 = \omega_i^{u_i x_i}$ to the product.
\end{proof}

Now, for any $x, y \in \mathbb{Z}_M$ and $t \in [0, 1]$, consider the distribution $N_{t, 1-t}(x, y)$, supported on $\mathbb{Z}_M$, to be the distribution given by letting $z \sim N_{t, 1-t}(x, y)$ have each $i \in [n]$ set to $z_i = x_i$ with probability $t$ and $z_i = y_i$ otherwise. Given a function $g: \mathbb{Z}_M \to \mathbb{C}$ and $t \in [0, 1]$, we define  
\begin{equation}\label{hehe2}
g_{t, 1-t}(x, y) = \mathbb{E}_{z \sim N_{t, 1-t}(x, y)}\big[g (z)\big] = \sum_{u \in \mathbb{Z}_M} \hat{g}(u) \prod_{i\in [n]} \Big(t \omega_i^{u_i x_i} + (1-t) \omega_i^{u_i y_i}\Big),
\end{equation}
where the second equation follows from arguments similar to the proof of Proposition \ref{hehe1}. 

Lastly, for any $\gamma > 0$, we let $\Delta^\gamma f$ be the   operator given by:
\begin{equation}\label{hehe4}
\Delta^\gamma f(x) = \sum_{u \in \mathbb{Z}_M} \hat{f}(u) (\# u)^\gamma \prod_{i\in [n]}\omega_i^{u_i x_i}.
\end{equation}

\subsection{Proof of the Robust Pisier's Inequality over Hypergrids}
Our proof follows the proof of \cite{NS02} (Theorem 2). For the robustness part, it adapts the proof strategy of \cite{canonne2021random}.
Let $M = \{m_1, \dots, m_n\}$.
We start with the following lemma:

\begin{lemma} \label{pisier_subpart} Let $f,g: \mathbb{Z}_M \to \mathbb{C}$ be two functions with $\mathbb{E}_{x \sim \mathbb{Z}_M}[f(x)] = 0$. Then we have
\begin{align*}
\sum_{u\ne 0 \in \mathbb{Z}_M} t^{(\# u) - 1} (\# u)^{\gamma + 1}\cdot  \hat{f}(u) \overline{\hat{g}(u)} 
=\frac{1}{1 - t}\cdot \mathbb{E}_{x, y \sim \mathbb{Z}_M} \left[ \overline{g_{t, 1-t} (x, y)}\cdot  \sum_{i\in [n]} \sum_{a \in \mathbb{Z}_{m_i}^*} \omega_i^{-a y_i} \omega_i^{a x_i} L_i \Delta^\gamma f(x) \right].
\end{align*}
\end{lemma}
\begin{proof} 
We work on the sum 
\begin{equation}\label{hehe3}
\overline{g_{t, 1-t}(x, y)} \cdot \sum_{i \in [n]} \sum_{a \in \mathbb{Z}_{m_i}^*} \omega_i^{-a y_i} \omega_i^{a x_i} L_i \Delta^\gamma f(x).
\end{equation}
Replacing $g_{t,1-t}(x,y)$ using the
  RHS of (\ref{hehe2}) and 
  $L_i\Delta^\gamma f(x)$ using (\ref{hehe5}) and (\ref{hehe4}),
  (\ref{hehe3}) becomes 
\begin{equation}\label{fourierexpand} 
\left[\sum_{v \in \mathbb{Z}_M} \overline{\hat{g}(v)} \prod_{j\in [n]}  \big(t \omega_j^{-v_j x_j } + (1-t) \omega_j^{-v_j y_j }\big) \right] 
 \left[\sum_{i\in [n]} \sum_{a \in \mathbb{Z}_{m_i}^*}  \omega_i^{-a y_i} \omega_i^{a x_i} \sum_{u \in \mathbb{Z}_M: u_i \neq 0} (\#u)^\gamma \hat{f}(u) \prod_{k\in [n]}\omega_k^{u_k x_k} \right].
\end{equation}
Next, upon expanding, (\ref{fourierexpand}) becomes:  
\begin{equation}\label{fourierexpand3}
\sum_{i\in [n]}
\sum_{u \in \mathbb{Z}_M: u_i \neq 0} \sum_{v \in \mathbb{Z}_M} \hat{f}(u) \overline{\hat{g}(v)} (\# u)^\gamma  \sum_{a \in \mathbb{Z}_{m_i}^*} \omega_i^{-a y_i} \omega_i^{a x_i} \left(\prod_{j \in [n]} \omega_j^{u_jx_j} \left(t \omega_j^{-v_j x_j } + (1-t) \omega_j^{-v_j y_j }\right)\right).
\end{equation}
Let us take the expectation of this expression over $x,y\sim \mathbb{Z}_M$. By linearity of expectation, we get  
\begin{align}\nonumber
&\sum_{i\in [n]}\sum_{u \in \mathbb{Z}_M: u_i \neq 0} \sum_{v \in \mathbb{Z}_M} \hat{f}(u) \overline{\hat{g}(v)} (\# u)^\gamma   \sum_{a \in \mathbb{Z}_{m_i}^*} A_{i,u,v,a}\\[0.5ex]\label{complicatedsum}
&\hspace{1cm}=\sum_{u\in \mathbb{Z}_M} \sum_{i\in [n]: u_i\ne 0}\sum_{v\in \mathbb{Z}_M}
\hat{f}(u) \overline{\hat{g}(v)} (\# u)^\gamma \sum_{a\in \mathbb{Z}_{m_i}^*} A_{i,u,v,a},
\end{align}
where 
$$
A_{i,u,v,a}:=\mathbb{E}_{x, y \sim \mathbb{Z}_M} \left[\omega_i^{-a y_i} \omega_i^{a x_i}  \left(\prod_{j\in [n]}  \omega_{j}^{u_jx_j} \left(t \omega_j^{-v_j x_j } + (1-t) \omega_j^{-v_j y_j }\right)\right) \right].
$$
Furthermore, $A_{i,u,v,a}$ can be written as a product of $n$ expectations. The $i$th expectation is given by 
$$
\mathbb{E}_{x_i,y_i\sim \mathbb{Z}_{m_i}}
\Big[\omega_i^{-a y_i} \omega_i^{a x_i} \omega_i^{u_i x_i} \left(t \omega_i^{-v_i x_i } + (1-t) \omega_i^{-v_i y_i }\right)\Big].
$$
The $i$th expectation can be written as the expectation of a sum of two terms. Given that $a\ne 0$, the expectation of the first term is always $0$.
The expectation of the second term is 
  $(1-t)$ when $v_i=u_i=-a$, and is $0$ otherwise.
Similarly, the $j$th expectation, for each $j\ne i$, is given by
$$
\mathbb{E}_{x_j,y_j\sim \mathbb{Z}_{m_j}}
\left[  \omega_j^{u_j x_j} \left(t \omega_j^{-v_j x_j } + (1-t) \omega_j^{-v_j y_j }\right)\right],
$$
which is $0$ when $u_j\ne v_j$.
When $u_j=v_j$, the expectation is $t$ if $u_j=v_j\ne 0$ and is $1$ if $u_j=v_j=0$.
\normalsize

Given this analysis, 
  we have that for any given $i$ and $u$ such that $u_i\ne 0$, there is a unique choice for $a$ (i.e., $a=-u_i\in \mathbb{Z}_{m_i}^*$) and $v$ (i.e., $v=u$) such that 
  $A_{i,u,v,a}$ is nonzero and is equal to
  $$(1-t)\cdot t^{\text{number of $j\ne i$ such that $u_j\ne 0$}}=(1-t)\cdot t^{(\#u)-1}.$$
As a result, (\ref{complicatedsum})
  can be simplified to 
$$
\sum_{u \in \mathbb{Z}_M}
\sum_{i\in [n]:u_i\ne 0}
\hat{f}(u) \overline{\hat{g}(u)} (\# u)^\gamma \cdot ({1-t})\cdot t^{(\#u)-1}
=(1-t)\cdot \sum_{u\ne 0\in \mathbb{Z}_M}
\hat{f}(u) \overline{\hat{g}(u)} (\# u)^{\gamma+1} t^{(\#u)-1},
$$
from which the lemma follows.
\end{proof}

Recall that given two functions $f,g: \mathbb{Z}_M\rightarrow \mathbb{C}$, we write
$$
\|f\|_s := \Big(\mathbb{E}_{x\sim \mathbb{Z}_M} \big[ |f(x)|^s\big]\Big)^{1/s}\quad\text{and}\quad \langle f,g\rangle := \mathbb{E}_{x\sim \mathbb{Z}_M}\big[ f(x)\overline{g(x)}\big].
$$
We are now ready to prove Theorem \ref{extended_robust_pisier}:

\begin{proof}[Proof of Theorem \ref{extended_robust_pisier}]
Let $\rho = 1 - 1/(n+1)$ and $q \in [1, \infty)$ such that $\frac{1}{s} + \frac{1}{q} = 1$. Given $f$, let  $g: \mathbb{Z}_M \to \mathbb{C}$ be a function with $\|g\|_q = 1$ satisfying $\langle T_\rho f, g \rangle = \|T_\rho f\|_s$. We have
$$(2\rho-1)^n \cdot  \|f\|_s \leq \|T_\rho f\|_s = \langle T_\rho f, g \rangle = \sum_{u\ne 0\in \mathbb{Z}_M} \rho^{(\# u)} \hat{f}(u) \overline{\hat{g}(u)},$$
where the last equation used the assumption that $\hat{f}(0)=0.$
Let $\gamma>0$ be a parameter, which will approach $0$ at the end of the proof.
By writing $$\rho^{(\# u)} = \frac{1}{\Gamma(1 + \gamma)} \int_0^\rho t^{(\# u) - 1}(\# u)^{\gamma + 1} \big(\log(\rho / t)\big)^\gamma dt,$$ for every $u\ne 0$ (in which case $\#u>0$), we have
\begin{equation} \label{equation23}
(2\rho-1)^n\cdot \|f\|_s\le 
 \frac{1}{\Gamma(1 + \gamma)} \int_0^\rho \left(\sum_{u\ne 0 \in \mathbb{Z}_M} t^{(\# u) - 1} (\# u)^{\gamma + 1} \hat{f}(u) \overline{\hat{g}(u)} \right) \big(\log (\rho / t)\big)^\gamma dt.
\end{equation}
By Lemma \ref{pisier_subpart}, the RHS of (\ref{equation23}) is
\begin{equation} \label{equation24}
 \frac{1}{\Gamma(1 + \gamma)}\int_0^\rho \frac{1}{1 - t} \cdot \mathbb{E}_{x, y \sim \mathbb{Z}_M} \left[ \overline{g_{t, 1-t} (x, y)} \sum_{i \in [n]} \sum_{a \in \mathbb{Z}_{m_i}^*} \omega_i^{-a y_i} \omega_i^{a x_i} L_i \Delta^\gamma f(x) \right] \cdot \big(\log (\rho / t)\big)^\gamma dt.
\end{equation}
Plugging in  $L_i f(x) = \sum_{b \in \mathbb{Z}_{m_i}} L_i^b f(x)$, this equals 
$$ \frac{1}{\Gamma(1 + \gamma)}\int_0^\rho \frac{1}{1 - t}\cdot \mathbb{E}_{x, y \sim \mathbb{Z}_M} \left[ \overline{g_{t, 1-t} (x, y)} \sum_{i\in [n]} \sum_{a \in \mathbb{Z}_{m_i}^*} \sum_{b \in \mathbb{Z}_{m_i}} \omega_i^{-a y_i} \omega_i^{a x_i} L_i^b \Delta^\gamma f(x) \right]\cdot \big(\log (\rho / t)\big)^\gamma dt.$$
Since this expression equals $||T_{\rho} f||_s$, which is real-valued, we can say it is less than its absolute value:
$$\leq \abs{ \frac{1}{\Gamma(1 + \gamma)}\int_0^\rho \frac{1}{1 - t}\cdot \mathbb{E}_{x, y \sim \mathbb{Z}_M} \left[ \overline{g_{t, 1-t} (x, y)} \sum_{i\in [n]} \sum_{a \in \mathbb{Z}_{m_i}^*} \sum_{b \in \mathbb{Z}_{m_i}} \omega_i^{-a y_i} \omega_i^{a x_i} L_i^b \Delta^\gamma f(x) \right]\cdot \big(\log (\rho / t)\big)^\gamma dt}$$
$$\leq  \frac{1}{\Gamma(1 + \gamma)}\int_0^\rho \frac{1}{1 - t}\cdot \abs{\mathbb{E}_{x, y \sim \mathbb{Z}_M} \left[ \overline{g_{t, 1-t} (x, y)} \sum_{i\in [n]} \sum_{a \in \mathbb{Z}_{m_i}^*} \sum_{b \in \mathbb{Z}_{m_i}} \omega_i^{-a y_i} \omega_i^{a x_i} L_i^b \Delta^\gamma f(x) \right]}\cdot \big(\log (\rho / t)\big)^\gamma dt.$$

Our next step is to obtain the following: 
\begin{align} 
\nonumber
&\abs{\mathbb{E}_{x, y \sim \mathbb{Z}_M} \left[ \overline{g_{t, 1-t} (x, y)} \sum_{i\in [n]} \sum_{a \in \mathbb{Z}_{m_i}^*} \sum_{b \in \mathbb{Z}_{m_i}} \omega_i^{-a y_i} \omega_i^{a x_i} L_i^b \Delta^\gamma f(x) \right]} \\ \label{equation25}
&\hspace{0.5cm}= \abs{\mathbb{E}_{x, y \sim \mathbb{Z}_M} \left[\overline{g_{t, 1-t} (x, y)} \sum_{i\in [n]}  \sum_{a \in \mathbb{Z}_{m_i}^*} \sum_{\substack{d \in \mathbb{Z}_{m_i}^* \\ (x, x^{(i) \to x_i + d}) \in G}}  \omega_i^{-a y_i} \omega_i^{a x_i} L_i^{x_i + d} \Delta^\gamma f(x) \cdot \left(1 - \omega_i^{ad}\right) \right]}
\\ \label{equation26}
   &\hspace{0.5cm}\le   \left(\mathbb{E}_{x, y \sim \mathbb{Z}_M} \left[  \abs{\sum_{i\in [n]}  \sum_{a \in \mathbb{Z}_{m_i}^*} \sum_{\substack{d \in \mathbb{Z}_{m_i}^* \\ (x, x^{(i) \to x_i + d}) \in G}} 
   (1-\omega_i^{ad}) \omega_i^{-a y_i} \omega_i^{a x_i} L_i^{x_i + d} \Delta^\gamma f(x)  }^s \right]\right)^{1/s}.
\end{align}

To prove (\ref{equation25}) we note that $L_i^{a}\Delta^\gamma f(x^{(i) \to a + d}) = - L_i^{a + d} \Delta^\gamma f(x^{(i) \to a})$. This is because $$L_i^{a}\Delta^\gamma f(x^{(i) \to a + d}) = 
\frac{\Delta^\gamma f(x^{(i)\rightarrow a+d})-\Delta^\gamma f(x^{(i)\rightarrow a})}{m_i}
=-L_i^{a+d}\Delta^\gamma f(x^{(i)\rightarrow a}).
$$
In the summation, we group terms corresponding to edges of $\mathcal{H}$ according to the orientation $G$:
\begin{align}\nonumber
&\mathbb{E}_{x, y \sim \mathbb{Z}_M} \left[ \overline{g_{t, 1-t} (x, y)} \sum_{i\in [n]} \sum_{a \in \mathbb{Z}_{m_i}^*} \sum_{b \in \mathbb{Z}_{m_i}} \omega_i^{-a y_i} \omega_i^{a x_i} L_i^b \Delta^\gamma f(x) \right] 
\\ \nonumber &= \mathbb{E}_{x \sim \mathbb{Z}_M} \Bigg[ \sum_{i\in [n]} \sum_{a \in \mathbb{Z}_{m_i}^*} \sum_{\substack{d \in \mathbb{Z}_{m_i} ^* \\ (x, x^{(i) \to x_i + d}) \in G}}  \mathbb{E}_{y \sim \mathbb{Z}_M} \bigg[  \overline{g_{t, 1-t} (x, y)} \cdot \omega_i^{-a y_i} \omega_i^{a x_i} L_i^{x_i + d} \Delta^\gamma f(x)\\ \nonumber & \hspace{4cm} + \overline{g_{t, 1-t} (x^{(i) \to x_i + d}, y)} \cdot\omega_i^{-a y_i} \omega_i^{a (x_i + d)} L_i^{x_i} \Delta^\gamma f(x^{(i) \to x_i + d})  \bigg]\Bigg]   \\ \nonumber 
   & = 
   \mathbb{E}_{x \sim \mathbb{Z}_M} \Bigg[ \sum_{i \in [n]} \sum_{a \in \mathbb{Z}_{m_i}^*} \sum_{\substack{d \in \mathbb{Z}_{m_i}^* \\ (x, x^{(i) \to x_i + d}) \in G}}  \mathbb{E}_{y \sim \mathbb{Z}_M} \bigg[\overline{g_{t, 1-t} (x, y)}\cdot  \omega_i^{-a y_i} \omega_i^{a x_i} L_i^{x_i + d} \Delta^\gamma f(x) \\ \label{equation27} & \hspace{4cm}
    - \overline{g_{t, 1-t} (x^{(i) \to x_i + d}, y)}\cdot \omega_i^{-a y_i} \omega_i^{a (x_i + d)} L_i^{x_i + d} \Delta^\gamma f(x)  \bigg]\Bigg].
\end{align}

Next, for a fixed $x$, $i$, $a$ and $d$, we have the following:
\begin{align}\nonumber
   & \mathbb{E}_{y \sim \mathbb{Z}_M} \left[\overline{g_{t, 1-t} (x, y)}\cdot  \omega_i^{-a y_i} \omega_i^{a x_i} L_i^{x_i + d} \Delta^\gamma f(x) - \overline{g_{t, 1-t} (x^{(i) \to x_i + d}, y)}\cdot \omega_i^{-a y_i} \omega_i^{a (x_i + d)} L_i^{x_i + d} \Delta^\gamma f(x)  \right]\\[0.7ex]
    & \hspace{0.6cm}= (1-\omega_i^{ad})\cdot  \mathbb{E}_{y \sim \mathbb{Z}_M} \left[\overline{g_{t, 1-t} (x, y)}\cdot   \omega_i^{-a y_i}\omega_i^{ax_i} L_i^{x_i + d} \Delta^\gamma f(x)    \right] \label{hehe6} .
\end{align}
This equation follows because when expanding the terms in $$\overline{g_{t, 1-t}(x, y)} =  \mathbb{E}_{z \sim N_{t, 1-t}(x, y)}\Big[\overline{g(z)}\Big] \quad\text{and}\quad
\overline{g_{t,1-t}(x^{(i)\rightarrow x_i+d},y)}= {\mathbb{E}_{z' \sim N_{t, 1-t}(x^{(i) \to x_i + d}, y)}\Big[\overline{g(z')}\Big]},$$
the LHS of (\ref{hehe6}) becomes
$$
\mathbb{E}_{y\sim \mathbb{Z}_M, z, z'}\Big[\overline{g(z)}\cdot \omega_i^{-a y_i} \omega_i^{a x_i} L_i^{x_i + d} \Delta^\gamma f(x)-\overline{g(z')}\cdot 
\omega_i^{-a y_i} \omega_i^{a (x_i+d)} L_i^{x_i + d} \Delta^\gamma f(x)\Big],
$$
where $z$ and $z'$ are drawn using the natural coupling that $z\sim N_{t,1-t}(x,y)$ and $z'$ is set to be $z$ if $z_i=y_i$ and $z'$ is set to be $z^{(i)\rightarrow x_i+d}$ if $z_i=x_i$.
Consider the following two cases:

\begin{flushleft}\begin{enumerate}
\item Either $z_i=y_i$ in which
  case $z'=z$ and thus,
  the contribution of the second term is always the contribution of the first term scaled by $-\omega_i^{ad}$; 

\item Or, $z_i = x_i$. In this case, both terms are independent of $y_i$ and thus, have an overall contribution of zero, as $\mathbb{E}_{y_i \sim \mathbb{Z}_{m_i} }[\omega_i^{-a y_i} ] = 0$ given that $a \neq 0$.
\end{enumerate}\end{flushleft}
This finishes the proof of (\ref{equation25}).

To obtain  (\ref{equation26}) from  (\ref{equation25}), we note that $\|\overline{g_{t, 1-t}}\|_q \leq \|g\|_q= 1$ and apply Hölder's inequality.

We proceed by substituting (\ref{equation26}) into (\ref{equation23}). For $\rho = 1 - {1}/{(n+1)}$, we have when $\gamma$ approaches $0$,  
$$\frac{1}{\Gamma(1 + \gamma)} \int_0^\rho \frac{\log (\rho / t)^\gamma}{1 - t} dt \lesssim \log n.$$ So we obtain:
$$(2\rho-1)^n  \|f\|_s \lesssim \log n \cdot \left(\mathbb{E}_{x, y \sim \mathbb{Z}_M} \left[  \abs{\sum_{i\in [n]}  \sum_{a \in \mathbb{Z}_{m_i}^*} \sum_{\substack{d \in \mathbb{Z}_{m_i}^* \\ (x, x^{(i) \to x_i + d}) \in G}} 
(1-\omega_i^{ad})\omega_i^{-a y_i} \omega_i^{a x_i} L_i^{x_i + d} \Delta^\gamma f(x)  }^s \right]\right)^{1/s}$$
As $\gamma \to 0$, the RHS approaches the desired quantity, while the LHS is independent of $\gamma$.
\end{proof}

\section{Proof of Lemma \ref{Lemma3.1}}\label{prooflemma9}

We prove 
Lemma \ref{Lemma3.1}  in this section. We follow the high-level strategy used in Section 3 of \cite{canonne2021random} but need to overcome a number of obstacles that are unique to hypergrids.

Let $t \in [n-1]$ be the parameter from Lemma \ref{Lemma3.1}. For this section, let $T$ denote a subset of $[n]$ of size $t$, and let $S$ denote a subset of $[n]$ of size $ t+1 $. 

The steps of the proof of Lemma \ref{Lemma3.1} are as follows. First, we apply the robust Pisier's inequality over hypergrids to connect the total variation distance $d_{TV}(p_{\overline{T}}, \mathcal{U})$, for a given $t$-subset $T$ of $[n]$, to the average $\sqrt{\text{outdegree}}$ of a collection of directed graphs over $\mathbb{Z}_K$, where $K:= M_{\overline{T}}=(m_i:i\notin T)$ defines a hypergrid $\mathbb{Z}_K$ of dimension $k:=n-t$.
Next we connect these graphs 
  with the bias vector $\mu(p_{|\rho})$ of either
  $\rho\sim \mathcal{D}(t,p)$ or 
  $\rho\sim \mathcal{D}(t+1,p)$ to finish the proof.

\subsection{Connecting Total Variation Distance to Directed Graphs} \label{tv_to_directed_graphs}

Fix any $t$-subset $T$ of $[n]$ and 
  let $K=M_{\overline{T}}$ of length $k=n-t$.
Let $\ell$ be a probability distribution over $\mathbb{Z}_K$. (Later on in the proof of Lemma \ref{Lemma3.1}, we will let $\ell$ be $p_{\overline{T}}$. We refer to $p_{\overline{T}}$ as $\ell$ in this subsection for notational convenience.)
Recall that $m=\max_{i\in [n]}m_i$ and thus, $m\ge \max_{i\in \overline{T}} m_i$.

Let $\mathcal{H}$ denote the undirected graph over $\mathbb{Z}_K$ consisting of undirected edges $\{x,x^{(i) \to b}\}$, for each $x \in \mathbb{Z}_K$, $i \in \overline{T}$, and $b \in \mathbb{Z}_{m_i}$ with $b \neq x_i$. Next, we assign weights to edges of $\mathcal{H}$ as follows.

\def\calH{\mathcal{H}}

\begin{definition}
An undirected edge $\{x, x^{(i) \to b}\} \in \mathcal{H}$  is a \textit{zero} edge if $\ell(x) = \ell(x^{(i) \to b})$. For each nonzero edge $\{x, x^{(i) \to b}\} \in \mathcal{H}$, let its \textit{weight} be defined as:
$$w(\{x, x^{(i) \to b} \}) := \frac{|\ell(x) - \ell(x^{(i) \to b})|}{\max \{\ell(x), \ell(x^{(i) \to b})\}}.$$
The weight of a nonzero edge is always in $(0, 1]$. A nonzero edge is called \textit{uneven} if its weight is at least $ {m}/({m+1})$. 
Otherwise (any nonzero edge with weight smaller than $ {m}/({m+1})$), we say it is an \textit{even} edge. An even edge is at \textit{scale} $\kappa$ for some integer $\kappa \geq 1$ if:
$$m^{-\kappa} < w\left(\{ x, x^{(i) \to b}\}\right) \leq m^{-\kappa + 1}.$$
\end{definition}

We partition edges of $\mathcal{H}$ to define three undirected graphs $\mathcal{H}^{[z]},
\mathcal{H}^{[u]},\mathcal{H}^{[e]}$ according to their weights: 
\begin{enumerate}
    \item  $\mathcal{H}^{[u]}$ (where $u$ stands for ``uneven''): Add all uneven edges of $\mathcal{H}$ to $\mathcal{H}^{[u]}$.
    \item  $\mathcal{H}^{[z]}$ (where $z$ stands for ``zero''): Add all zero edges of $\calH$ to  $\mathcal{H}^{[z]}$; and
    \item   $\mathcal{H}^{[e]}$
    (where $e$ stands for ``even''): Add all even edges of $\calH$  to $\mathcal{H}^{[e]}$.
\end{enumerate}
Next we assign orientations to edges in 
  $\calH^{[u]}$ and $\calH^{[z]}$
  to obtain directed graphs 
  $G^{[u]}$ and $G^{[z]}$:
\begin{flushleft}\begin{enumerate}
\item $G^{[u]}$: For each uneven edge $\{x, y\} \in \mathcal{H}^{[u]}$, orient the edge from $x$ to $y$ if $\ell(x) >\ell(y)$ and from $y$ to $x$ if $\ell(y)>\ell(x)$. Note that 
$\ell(x)\ne \ell(y)$ since it is not a zero edge so the directions are well defined. 
\item $G^{[z]}$: Orient each zero edge $\{x, y\} \in \mathcal{H}^{[z]}$  arbitrarily.
\end{enumerate}\end{flushleft}
Orientations of even edges are  trickier. Notably our construction below is significantly different from that of \cite{canonne2021random}. 
We partition and orient even edges
  into directed graphs
  $G^{[\kappa]}$ for each $\kappa\ge 1$ and $G^{[r]}$,
  where each $G^{[\kappa]}$
  contains orientations of a subset of 
  even edges at scale $\kappa$
  and $G^{[r]}$ (where $r$ stands for ``remaining'') contains orientations of even edges not included in $G^{[\kappa]}$'s:
\begin{flushleft}\begin{enumerate}
    \item  $G^{[\kappa]}$, for each $\kappa \geq 1$: 
    First we define $\calH^{[\kappa]}$ to be the undirected graph over $\mathbb{Z}_K$ that includes 
    all even edges $\{x,y\}\in \calH$ of scale $\kappa$ if $y=x^{(i)\to b}$ for some $i$ and $b$ (so the edge is along the $i$-th direction) satisfies that \emph{neither $x$ nor $y$ has any  outgoing edges in $G^{[u]}$ along the $i$-th direction}.\vspace{0.03cm}
    
    We then orient edges in $\calH^{[\kappa]}$ to obtain the directed graph $G^{[\kappa]}$ as follows. For each $\kappa\ge 1$, find an ordering of vertices in $\mathbb{Z}_K$ as a bijection $\smash{\rho_\kappa: \mathbb{Z}_K \to [\prod_{j\in \overline{T}} m_j]}$ (i.e., $x$ is the $\rho_\kappa(x)$-th vertex in the ordering) such that $\rho_\kappa$ satisfies the following property: For each $i \in [\prod_{j\in \overline{T}} m_j-1]$, the degree of $\smash{\rho_\kappa^{-1}(i)}$ is the largest out of all vertices in the subgraph of $\smash{\mathcal{H}^{[\kappa]}}$ induced by $\smash{\{\rho_\kappa^{-1}(j) : j \geq i \}}$.\vspace{0.03cm}
    
     Starting with $i=1$, one can construct such a bijection $\rho_\kappa$ by deleting vertices one at a time from $\smash{\calH^{[\kappa]}}$, at each step deleting the vertex with the largest degree in the remaining undirected graph, making it $\rho_\kappa(i)$ and setting $j=j+1$. Ties can be broken arbitrarily.\vspace{0.03cm}
    
     We now use $\rho_\kappa$ to orient  the edges in $\calH^{[\kappa]}$ to obtain the directed graph $G^{[\kappa]}$: For each undirected edge $\{x, y\}$ in $ \mathcal{H}^{[\kappa]}$, orient the edge from $x$ to $y$ if $\rho_\kappa(x) < \rho_\kappa(y)$, and orient the edge from $y$ to $x$ otherwise. This ensures that every directed edge $(x, y) \in G^{[\kappa]}$ satisfies $\rho_\kappa(x) < \rho_\kappa(y)$.\vspace{0.03cm}
    
    \item  $G^{[r]}$: For every even edge $\{x,y\}$ in $\calH$ that was not included in $\calH^{[\kappa]}$'s (which means that one of its vertices  has at least one outgoing edge in $G^{[u]}$ along the same direction), add $(x,y)$ to $G^{[r]}$ if $x$ has at least one outgoing edge
    in $G^{[u]}$ along the same direction and add $(y,x)$ to $G^{[r]}$ if $y$ has at least one outgoing edge along the same direction, breaking ties arbitrarily.\vspace{0.03cm}
\end{enumerate}\end{flushleft}

In the analysis proving Lemma \ref{Lemma3.1}, we will utilize the following fact about the directed graph $G^{[\kappa]}$ (this fact, over hypercubes, can be found in \cite{canonne2021random}):

\begin{lemma}\label{Lemma3.4}
Let $U$ be a set of vertices in $\mathbb{Z}_K$ and let $v \in \mathbb{Z}_K \setminus U$. If the outdegree of every vertex $u \in U$ in $G^{[\kappa]}$ is bounded from above by a positive integer $g$, then the number of directed edges $(u, v)$ from a vertex $u \in U$ to $v$ in $G^{[\kappa]}$ is also at most $g$.
\end{lemma}

\begin{proof}
Consider the vertex $s$ that is ranked the highest (i.e., smallest value) in $\rho_\kappa$ among $U\cup \{v\}$.
If $s$ is $v$, then all undirected edges between $U$ and $v$ are oriented from $v$ to $U$ so the number of directed edges $(u, v)$ is $0$.
If $s\in U$, then the assumption implies that the subgraph of $\calH^{[\kappa]}$ induced by $U\cup \{v\}$ has maximum degree at most $g$, including the degree of $v$, from which the lemma follows trivially.
\end{proof}

With $G^{[z]}, G^{[u]}$, $G^{[\kappa]}$ for each $\kappa \geq 1$, and $G^{[r]}$ defined, we then define $G$ to be the union of these directed graphs, which is an orientation of $\mathcal{H}$ over $\mathbb{Z}_K$. 

We now apply Theorem \ref{extended_robust_pisier} (the robust Pisier's Inequality over hypergrids) in a way that connects $d_{TV}(\ell, \mathcal{U})$ to the directed edges of $G$. To do so, define the function $f: \mathbb{Z}_K \to [-1, \infty)$ as follows: For each $y\in \mathbb{Z}_K$,
\begin{equation}
    f(y) = \left(\prod_{j\in \overline{T}} m_{j} \right) \cdot \ell(y) - 1.
\end{equation}
Note that $\mathbb{E}_{y \sim \mathbb{Z}_K}[f(y)] = 0$.
Setting $s=1$, the left-hand side of the robust Pisier's inequality gives  $$\mathbb{E}_{y \sim \mathbb{Z}_K}\big[\abs{f(y)}\big] = 2\cdot d_{TV}(\ell, \mathcal{U}).$$
We use the robust Pisier's inequality to prove the following lemma:
\begin{lemma} \label{bernstein_bound} 
For any probability distribution $\ell$ over $\mathbb{Z}_K$, we have
$$\frac{d_{TV}(\ell, \mathcal{U})}{m^{1.5} \log^2n} \lesssim \mathbb{E}_{x \sim \mathbb{Z}_K}\left[ \sqrt{\sum_{i \in \overline{T}} \sum_{\substack{b \in \mathbb{Z}_{m_i} \\ (x, x^{(i) \to b}) \in G}} \left(L_i^b  f(x)\right)^2 } \right].$$
\end{lemma}

\begin{proof}
A direct application of the robust Pisier's inequality  (Theorem \ref{extended_robust_pisier}) with $s = 1$ gives
$$\frac{d_{TV}(\ell, \mathcal{U})}{\log n} \lesssim  \mathbb{E}_{x, y \sim \mathbb{Z}_K} \left[ \abs{\sum_{i\in \overline{T}}  \sum_{a \in \mathbb{Z}_{m_i}^*} \sum_{\substack{d \in \mathbb{Z}_{m_i}^* \\ (x, x^{(i) \to x_i + d}) \in G}} \hspace{-0.3cm}(1-\omega_i^{ad})\omega_i^{-a y_i} \omega_i^{a x_i} L_i^{x_i + d} f(x)  } \right]$$
For convenience, we write $B_{x,i,a,d}$ to denote 
$$
B_{x,i,a,d}:= (1-\omega_i^{ad})\omega_i^{ax_i}L_i^{x_i+d}f(x).
$$
Then the RHS of the inequality  above becomes
\begin{equation}\label{haha5}
\mathbb{E}_{x, y \sim \mathbb{Z}_K} \left[ \abs{\sum_{i\in \overline{T}}  \sum_{a \in \mathbb{Z}_{m_i}^*} \sum_{\substack{d \in \mathbb{Z}_{m_i}^* \\ (x, x^{(i) \to x_i + d}) \in G}} \hspace{-0.3cm} \omega_i^{-a y_i}B_{x,i,a,d}} \right].
\end{equation}

Let $R(z)$ and $I(z)$ denote the real and imaginary parts of a complex number $z \in \mathbb{C}$. Then we have
\begin{align*}
&\abs{\sum_{i\in \overline{T}}  \sum_{a \in \mathbb{Z}_{m_i}^*} \sum_{\substack{d \in \mathbb{Z}_{m_i}^* \\ (x, x^{(i) \to x_i + d}) \in G}} \hspace{-0.3cm} \omega_i^{-a y_i}B_{x,i,a,d}}\\
&\hspace{0.2cm}\le 
 \abs{\sum_{i\in \overline{T}}  \sum_{a \in \mathbb{Z}_{m_i}^*} \sum_{\substack{d \in \mathbb{Z}_{m_i}^* \\ (x, x^{(i) \to x_i + d}) \in G}}    R(\omega_i^{-a y_i}) R(B_{x,i,a,d})  } + \abs{\sum_{i\in \overline{T}}  \sum_{a \in \mathbb{Z}_{m_{i}}^*} \sum_{\substack{d \in \mathbb{Z}_{m_{i}}^* \\ (x, x^{(i) \to x_i + d}) \in G}}   R(\omega_i^{-a y_i}) I(B_{x,i,a,d})  } \\
\label{real-imaginary-equation}
&\hspace{0.8cm}+ \abs{\sum_{i\in \overline{T}}  \sum_{a \in \mathbb{Z}_{m_i}^*} \sum_{\substack{d \in \mathbb{Z}_{m_{i}}^* \\ (x, x^{(i) \to x_i + d}) \in G}}    I(\omega_i^{-a y_i}) R(B_{x,i,a,d})  }  + \abs{\sum_{i\in \overline{T}}  \sum_{a \in \mathbb{Z}_{m_{i}}^*} \sum_{\substack{d \in \mathbb{Z}_{m_{i}}^* \\ (x, x^{(i) \to x_i + d}) \in G}}    I(\omega_i^{-a y_i}) I(B_{x,i,a,d})  }
\end{align*}
and we can now analyze real-valued random variables.
We analyze the first of the four terms:
\begin{equation}\label{haha6}
\mathbb{E}_{y \sim \mathbb{Z}_K} \left[  \abs{\sum_{i\in \overline{T}}  \sum_{a \in \mathbb{Z}_{m_{i}}^*} \sum_{\substack{d \in \mathbb{Z}_{m_{i}}^* \\ (x, x^{(i) \to x_i + d}) \in G}}    R(\omega_i^{-a y_i}) R(B_{x,i,a,d})  } \right],
\end{equation}
noting that the same analysis will apply to each of the other terms.
Define the random variable 
$$
X_i =R(\omega^{-a y_i}) \sum_{a \in \mathbb{Z}_{m_i}^*} \sum_{\substack{d \in \mathbb{Z}_{m_i}^* \\ (x, x^{(i) \to x_i + d}) \in G}}     R(B_{x,i,a,d})
$$ 
and note that the expectation of $X_i$ is $0$ over $y\sim \mathbb{Z}_K$. Let 
\begin{align*}
t &:= 100 \log n \cdot \sqrt{\sum_{i\in \overline{T}}  \left(\sum_{a \in \mathbb{Z}_{m_{i}}^*} \sum_{\substack{d \in \mathbb{Z}_{m_{i}}^* \\ (x, x^{(i) \to x_i + d}) \in G}}   R(B_{x,i,a,d})\right)^2}\\[1ex]
&\le O(\log n)\cdot \sqrt{\sum_{i\in \overline{T}} m^2 \sum_{a\in \mathbb{Z}^*_{m_i}}\sum_{\substack{d\in \mathbb{Z}_{m_i}^*\\ (x,x^{(i)\rightarrow x_i+d})}\in G}
\big(R(B_{x,i,a,d})\big)^2
}\\[0.6ex]
&\le O(m^{1.5}\log n)\cdot
\sqrt{\sum_{i\in \overline{T}}
\sum_{\substack{b\in \mathbb{Z}^*_{m_i}\\ (x,x^{(i)\rightarrow b)}\in G}}\left(L_i^bf(x)\right)^2}. 
\end{align*}
Bernstein's inequality gives us that 
$$\mathbb{P}_{y \sim \mathbb{Z}_K}\left( \abs{\sum_{i\in \overline{T}}  \sum_{a \in \mathbb{Z}_{m_{i}}^*} \sum_{\substack{d \in \mathbb{Z}_{m_{i}}^* \\ (x, x^{(i) \to x_i + d}) \in G}}   R(\omega^{-a y_i}) R(B_{x,i,a,d})  } \geq t  \right) \leq \frac{1}{n^{10}}.$$
As a result, we know that (\ref{haha6}) is at most 
\begin{equation}\label{bound-expectation-Bernstein}
\left(1-\frac{1}{n^{10}}\right)\cdot t+\frac{1}{n^{10}}\cdot nt< 2t.
\end{equation}
The same series of steps applies to the other three terms and the lemma follows. 
\end{proof}

Letting $G'$ be the directed graph that contains the union of edges in $G^{[u]},G^{[r]}$ and $G^{[\kappa]}$, $\kappa\ge 1$, but not those in $G^{[z]}$, we can replace the RHS of the Lemma \ref{bernstein_bound} inequality with
$$
\mathbb{E}_{x \sim \mathbb{Z}_K} \left[  \sqrt{ \sum_{i\in \overline{T}} \sum_{\substack{b \in \mathbb{Z}_{m_{i}} \\ (x, x^{(i) \to b}) \in G}} \big(L_i^{b} f(x) \big)^2 } \right] 
= \mathbb{E}_{x \sim \ell}\left[ \sqrt{\sum_{i \in \overline{T}} \sum_{\substack{b \in \mathbb{Z}_{m_{i}} \\ (x, x^{(i) \to b}) \in G'}} \left(\frac{L_i^b \ell(x)}{\ell(x)}\right)^2} \right].
$$

The next lemma connects the quantity in the expectation to the outdegree of $x$ in  $G^{[u]}$ and $G^{[\kappa]}$.

\begin{lemma} \label{Lemma3.5}
For every $x \in \mathbb{Z}_K$, we have 
$$\sum_{i \in \overline{T}} \sum_{\substack{b \in \mathbb{Z}_{m_i} \\ (x, x^{(i) \to b}) \in G'}} \left(\frac{\ell(x) - \ell(x^{(i) \to b})}{\ell(x)}\right)^2 \leq m^3 \cdot \text{outdeg}(x, G^{[u]}) +  \sum_{\kappa \geq 1} 4m^{-2 \kappa + 4} \cdot \text{outdeg}(x, G^{[\kappa]}).$$
\end{lemma}

\begin{proof}
Each edge $(x, x^{(i) \to b}) \in G'$  lies in $G^{[u]},G^{[r]}$ or $G^{[\kappa]}$ for some $\kappa \geq 1$. If $(x, x^{(i) \to b})$ is in $G^{[u]}$, then by the orientation of edges in $G^{[u]}$, we have $\ell(x) > \ell(x^{(i) \to b})$, which implies that the contribution of each such edge to the sum on the LHS is at most $1$.

Next, for each $(x, x^{(i) \to b})$ is in $G^{[\kappa]}$ for some $\kappa \geq 1$, since it is even, we have $\ell(x), \ell(x^{(i) \to b}) > 0$, since otherwise it is a zero edge or uneven edge. Since $w(\{x, x^{(i) \to b} \}) < {m}/({m+1})$, we have:
$$\frac{\max\{ \ell(x), \ell(x^{(i) \to b})\}}{\min\{ \ell(x), \ell(x^{(i) \to b})\}} < m+1.$$ 
Consequently, we have 
$$\frac{|\ell(x) - \ell(x^{(i) \to b}) |}{\ell(x)} \leq w(\{x, x^{(i) \to b} \}) \cdot \frac{\max\{ \ell(x), \ell(x^{(i) \to b})\}}{\min\{ \ell(x), \ell(x^{(i) \to b})\}} \leq (m+1) \cdot w(\{x, x^{(i) \to b} \}) \leq 2 m^{-\kappa + 2}.$$
Therefore, the contribution of each such edge to the sum on the LHS is at most $4 m^{-2\kappa + 4}$.

Lastly, assume that  $(x, x^{(i) \to b})$ is in $G^{[r]}$. By our construction, this implies that $(x,x^{(i)\to b})$ is an even edge and  there exists a $c \in \mathbb{Z}_{m_i} \setminus \{x_i\}$ such that $(x, x^{(i) \to c}) \in G^{[u]}$. This also implies that $(x, x^{(i) \to b})$ is a level $\kappa$ edge, for some $\kappa \geq 1$. 
By a similar argument as above, we have 
$$\left(\frac{\ell(x) - \ell(x^{(i) \to b}) }{\ell(x)}\right)^2 \leq \left(w(\{x, x^{(i) \to b} \}) \cdot \frac{\max\{ \ell(x), \ell(x^{(i) \to b})\}}{\min\{ \ell(x), \ell(x^{(i) \to b})\}}\right)^2 \le 
m^2,
$$
which is at most 
$m^2\cdot (\text{outdegree of } x \text{ in } G^{[u]} \text{ with respect to edges in the \textit{i}-th direction})$. Summing over all $i \in \overline{T}$ and $b \in \mathbb{Z}_{m_i}\setminus \{x_i\}$ such that $(x, x^{(i) \to b}) \in G^{[r]}$, we find:
$$\sum_{i \in \overline{T}} \sum_{\substack{b \in \mathbb{Z}_{m_{i}} \\ (x, x^{(i) \to b}) \in G^{[r]}}} \left(\frac{\ell(x) - \ell(x^{(i) \to b})}{\ell(x)}\right)^2 \leq (m-1)m^2 \cdot outdeg(x, G^{[u]}).$$

The lemma follows by combining the analysis for edges in  $G^{[u]}, G^{[r]}$ and $G^{[\kappa]}$.
\end{proof}

Finally, we connect $d_{TV}(\ell,\cal{U})$ with the expected $\sqrt{\text{outdegree}}$ of $x\sim \ell$:

\begin{lemma} \label{Lemma3.6}
Letting $\beta = d_{TV}(\ell, \mathcal{U})$, one of the following two conditions must hold:
\begin{enumerate}
    \item Either the directed graph $G^{[u]}$ of uneven edges satisfies:
    $$\mathbb{E}_{x \sim \ell} \left[\sqrt{\text{outdeg}(x, G^{[u]})} \right] \gtrsim \frac{\beta}{m^3 \log^2 n}.$$
    \item Or, there exists a $\kappa \in [10\log(n m / \beta)]$ such that the directed graph $G^{[\kappa]}$ satisfies:
    $$\mathbb{E}_{x \sim \ell} \left[\sqrt{\text{outdeg}(x, G^{[\kappa]})}\right] \gtrsim \frac{m^{\kappa} \beta}{m^{3.5}\cdot  \log^2 n  \cdot \log (nm/\beta)}.$$
\end{enumerate}
\end{lemma}

\begin{proof}
It follows from Lemmas \ref{bernstein_bound} and \ref{Lemma3.5} that
\begin{align*}
\hspace{-0.8cm}\frac{\beta}{m^{1.5} \log^2 n} &\lesssim \mathbb{E}_{x \sim \ell} \left[\sqrt{m^3 \cdot \text{outdeg}(x, G^{[u]}) +  \sum_{\kappa \geq 1} 4m^{-2 \kappa + 4} \cdot \text{outdeg}(x, G^{[\kappa]})}\right] \\
&\leq m^{1.5} \cdot \mathbb{E}_{x \sim \ell} \left[ \sqrt{\text{outdeg}(x, G^{[u]})} \right] + \sum_{\kappa = 1}^{10\log(n m / \beta)}  2 m^{-\kappa + 2} \cdot \mathbb{E}_{x \sim \ell} \left[ \sqrt{\text{outdeg}(x, G^{[\kappa]})} \right] + o\left(\frac{\beta}{m^{1.5} \log^2n}\right).
\end{align*}
which used the fact that the degrees are always bounded by $n (m-1)$. The lemma follows.
\end{proof}

\subsection{Separating into Cases} \label{section_bucketing}

For each $t$-subset $T$ of $[n]$, let $\alpha(T) = d_{TV}(p_{\overline{T}}, \mathcal{U})$. Note that the $\alpha$ in the statement of Lemma \ref{Lemma3.1} can be written as $\alpha = \mathbb{E}_{T \sim \mathcal{S}(t)} [\alpha(T)]$. For each $T$, take $p_{\overline{T}}$ as $\ell$ in the previous subsection to partition undirected edges in the undirected graph $\mathcal{H}(T)$ over $\mathbb{Z}_K$ with $K=M_{\overline{T}}$ into $\smash{\mathcal{H}^{[z]}(T)}$ (zero edges), $\smash{\mathcal{H}^{[u]}(T)}$ (uneven edges), and $\smash{\mathcal{H}^{[\kappa]}(T)}$ (even edges at scale $\kappa \geq 1$). Orient these edges as in the previous subsection to obtain directed graphs $G^{[u]}(T)$ and $G^{[\kappa]}(T)$. By applying Lemma \ref{Lemma3.6} on $\ell=p_{\overline{T}}$, we conclude that one of the following two conditions holds for either $G^{[u]}(T)$ or one of the graphs $G^{[\kappa]}(T)$, $\kappa \in  [10\log (n m / \alpha(T))].$ These cases mirror the two cases in the hypercube setting from \cite{canonne2021random}.
 
 Before stating the two cases, note that since $\alpha(T) \in [0, 1]$, there exists a $\zeta > 0$ such that with probability at least $\zeta$ over $T \sim \mathcal{S}(t)$, 
 $$\alpha(T) \gtrsim \frac{\alpha}{\zeta \log (1/\alpha)}.$$
 Therefore, one of the following cases must hold.
 
 \begin{flushleft}
 \textbf{Case 1:} With probability at least $\zeta/2$ over  $T \sim S(t)$, the directed graph $G^{[u]}(T)$ of   $p_{\overline{T}}$  satisfies
 $$\mathbb{E}_{x \sim p_{\overline{T}}} \left[ \sqrt{\text{outdeg}(x, G^{[u]}(T))}\right] \gtrsim \frac{\alpha}{\zeta \log (1/\alpha)\cdot m^3 \log^2 n}.$$
 
 Since the out-degree is always between $0$ and $n(m-1)$, there exist two parameters $d \in [n (m-1)]$ and $\xi > 0$ such that with probability $\zeta/(2 \log (nm))$ over the draw of $T \sim S(t)$, we have
\begin{equation}
\Pr_{x \sim p_{\overline{T}}} \left[ d \leq \text{outdeg}(x, G^{[u]}(T)) \leq 2d \right] \geq \xi
\end{equation}
 and $\xi$ satisfies
\begin{equation}
\label{equation32}
\sqrt{d} \cdot \xi \gtrsim \frac{\alpha}{\zeta \log (1/\alpha)\cdot 
 m^3 \log^2 n \cdot \log (nm) }. 
\end{equation} 
 
 \textbf{Case 2:} There exists a  $\kappa \in [O(\log (nm / \alpha))]$ (using $\zeta \leq 1$) such that with probability at least
 $$\Omega\left( \frac{\zeta}{\log (nm/
 \alpha)}\right)$$ over $T \sim S(t)$, the directed graph $G^{[\kappa]}(T)$ of even edges at scale $\kappa$ of $p_{\overline{T}}$  satisfies
  $$\mathbb{E}_{x \sim p_{\overline{T}}} \left[ \sqrt{\text{outdeg}(x, G^{[\kappa]})(T)}\right] \gtrsim \frac{m^\kappa\alpha}{\zeta \log (1/\alpha)\cdot m^{3.5} \log^2 n \cdot \log (nm/\alpha)}.$$
  
  Using a bucketing argument, there exist $d \in [n(m-1)]$ and $\xi > 0$ such that with probability $$
  \Omega\left(\frac{\zeta}{\log (nm)\cdot  \log (nm/\alpha)}\right)$$ over the draw of $T \sim S(t)$, we have $$\mathbb{P}_{x \sim p_{\overline{T}}} \left[ d \leq \text{outdeg}(x, G^{[\kappa]}(T)) \leq 2d \right] \geq \xi$$
 and $\xi$ satisfies \begin{equation}\sqrt{d} \cdot \xi \gtrsim \frac{ m^\kappa\alpha}{\zeta\log (1/\alpha)\cdot m^{3.5} \log^2 n \cdot \log(nm/\alpha) \cdot \log (nm) }.
 \label{equation33}
 \end{equation}
 \end{flushleft}

\subsection{From Directed Graphs to the Bias Vector}\label{sec:haha1}

Let $R$ be a subset of $[n]$ (which will be either a $t$-subset $T$ of $[n]$ or a $(t+1)$-subset $S$ of $[n]$ in the rest of the section). Let $J=M_{{R}}$.
Given a distribution $\ell$ over $\mathbb{Z}_J$, $i\in R$ and 
  $c,d\in \mathbb{Z}_{m_i}$,
  recall the definition of the bias $\smash{\mu^{c,d}_i(\ell)}$ from Definition  \ref{mu-c_definition}:
$$
\mu_i^{c,d}(\ell)=\frac{\Pr_{x\sim \ell}[x_i=c]-\Pr_{x\sim\ell}[x_i=d]}
{\Pr_{x\sim \ell}[x_i=c]+\Pr_{x\sim\ell}[x_i=d]}
$$  
with $\mu_i^{c,d}(\ell)=0$ with $\Pr_{x\sim \ell}[x_i=c]=\Pr_{x\sim\ell}[x_i=d]=0$.

In this subsection we connect directed graphs defined in Section \ref{tv_to_directed_graphs} and \ref{section_bucketing} to 
  biases of restrictions $p_{|\rho}$ of
  $p$ when $\rho\sim \mathcal{D}(t+1,p)$.
Consider a distribution $p$ supported on $\mathbb{Z}_M$ and let $t \in [n-1]$. 

Let $\pi = (\pi(1), \dots, \pi(t+1))$ be an ordered sequence of $t+1$ distinct indices from $[n]$. We let $S(\pi)$ denote the corresponding $(t+1)$-subset $\{\pi(1), \dots, \pi(t+1)\}$. 
\begin{definition}
Given $\pi$ and $y \in \mathbb{Z}_M$, define a restriction $\rho(\pi, y) \in \bigtimes_{i = 1}^n (\mathbb{Z}_{m_i} \cup \{* \})$ as
$$\rho(\pi, y)_i = \begin{cases}
  * ~~~~ i = \pi(j) \text{ for some } j \in [t+1]\\
  y_i ~~ \text{ otherwise}.
\end{cases}$$
\end{definition}

We will also consider sequences $\tau = (\tau(1), \dots, \tau(t))$ of $t$ (instead of $t+1$) distinct indices from $[n]$. 
For such a $\tau$, the corresponding set $S(\tau)$ and the restriction $\rho(\tau, y)$ given by $y \in \mathbb{Z}_M$ are defined similarly.

As in \cite{canonne2021random}, we use that the following is an equivalent way of drawing $\rho\sim\mathcal{D}(t + 1, p)$:
\begin{flushleft}\begin{enumerate}
\item First, sample a sequence of $t+1$ random indices $\pi=(\pi(1),\ldots,\pi(t+1))$ uniformly from $[n]$
  without replacements (so the set $S(\pi)$ can be viewed equivalently as drawn from $\mathcal{S}(t+1)$).
\item Then, sample $y \sim p$.
\item Finally, return $\rho=\rho(\pi,y)$.
\end{enumerate}\end{flushleft}
We will use $(\pi,y)\sim \mathcal{D}'(t+1,p)$ to denote the sampling of $(\pi,y)$ as above, with the understanding that $\rho(\pi,y)$ is distributed the same as $\mathcal{D}(t+1,p)$.
Similarly, consider sampling $\rho\sim \mathcal{D}(t,p)$ equivalently according to the following procedure:
\begin{flushleft}\begin{enumerate}
\item First, sample a sequence of $t$ random indices $\tau=(\tau(1),\dots,\tau(t ))$ uniformly from $[n]$
  without replacements (so the set $S(\tau)$ can be viewed equivalently as drawn from $\mathcal{S}(t )$).
\item Then, sample ${y} \sim p$.
\item Finally, return $\rho=\rho(\tau,{y})$.
\end{enumerate}\end{flushleft}
Similarly we will write $(\tau,y)\sim \mathcal{D}'(t,p)$ to denote the sampling of $(\tau,y)$ as above.

Fixing any $i\in [t+1]$, we 
let $\pi_{-i}$ denote the length-$t$ sequence obtained from $\pi$ after removing its $i$-th entry. An important observation is that $\rho\sim \mathcal{D}(t,p)$ can also be drawn as follows:
\begin{flushleft}\begin{enumerate}
\item First, sample a sequence of $t+1$ random indices $\pi=(\pi(1),\ldots,\pi(t+1))$ uniformly from $[n]$
  without replacements and set $\tau=\pi_{-i}$. (Note that $i$ is a fixed index in $[t+1]$.)
\item Then, sample $y \sim p$.
\item Finally, return $y$ and  $\rho=\rho(\tau,y)=\rho(\pi_{-i},y)$.
\end{enumerate}\end{flushleft}

Given a $t$-subset $T$ of $[n]$ with $K=M_{\overline{T}}$, we will 
  use 
$$(z,i,b)\in \mathbb{Z}_K\times \overline{T}\times \mathbb{Z}_{m_i}$$
to denote directed edges 
  over $\mathbb{Z}_K$: $(z,i,b)$ means the directed edge $(z,z^{(i)\to b})$ so we can talk about, e.g., whether
  $(z,i,b)\in G^{[u]}(T)$ and whether
  $(z,i,b)\in G^{[\kappa]}(T)$.
(Note that for notational convenience, we allow $b$ to be $z_i$, in which case $(z,i,b)$ can never be an edge in these directed graphs.)
As an example, let $y \in \mathbb{Z}_M$, $\pi$ be a $(t+1)$-sequence of distinct elements in $[n]$, and $b \in \mathbb{Z}_{m_{\pi(i)}}$ for some $i\in [t+1]$. Then, 
$\smash{(y_{\overline{S(\pi_{-i})}}, \pi(i), b)}$ denotes the edge from $\smash{y_{\overline{S(\pi_{-i})}}}$ to $y'$ where $y'$ satisfies $\smash{y'_{\pi(i)} = b}$ and $\smash{y'_{\overline{S(\pi)}} = y_{\overline{S(\pi)}}}$.

The following lemma connects the directed graphs to biases of restrictions of $p$:

\begin{lemma} \label{Lemma3.7}
Let $\pi$ be a $(t+1)$-sequence of distinct indices  and $y \in \mathbb{Z}_M$. For $i \in [t+1]$ and $b \in \mathbb{Z}_{m_{\pi(i)}}$, 
\begin{equation}
\begin{aligned}
\left|\mu^{b,y_{\pi(i)}}_{\pi(i)}(p_{| \rho(\pi, y)})\right| &\geq \frac{m}{2(m+1)} \cdot \mathbbm{1} \left\{ \left(y_{\overline{S(\pi_{-i})}}, \pi(i), b\right) \in G^{[u]}(S(\pi_{-i})) \right\} \\[0.5ex]
&\hspace{0.8cm}+ \sum_{\kappa \geq 1} \frac{1}{2m^\kappa} \cdot \mathbbm{1} \left\{ \left(y_{\overline{S(\pi_{-i})}}, \pi(i), b\right) \in G^{[\kappa]}(S(\pi_{-i})) \right\}.
\end{aligned}\label{hehe100}
\end{equation}
\end{lemma}
\begin{proof}
We let $\ell$ denote $p_{|\rho(\pi, y)}$ and $c$ denote $y_{\pi(i)}$. Writing
$$\text{Pr}^b = \Pr_{x \sim p}\left[x_{\pi(i)} = b, x_{\overline{S(\pi)}} = y_{\overline{S(\pi)}}\right]\quad\text{and} \quad \text{Pr}^c = \Pr_{x \sim p}\left[x_{\pi(i)} = c, x_{\overline{S(\pi)}} = y_{\overline{S(\pi)}}\right],$$
we have that the LHS of (\ref{hehe100}) is $$\left|\mu^{b,c}_{\pi(i)}(\ell)\right| =\left|\frac{ \text{Pr}^b - \text{Pr}^c }{\text{Pr}^b + \text{Pr}^c}\right|.$$
Let $z$ be the string with $z_{\overline{S(\pi)}} = y_{\overline{S(\pi)}}$ and $z_{\pi(i)} = c$, 
and  $z'$ be with $z'_{\overline{S(\pi)}} = y_{\overline{S(\pi)}}$ and $z'_{\pi(i)} = b$. Then
$$w({z, z'}) = \frac{|\text{Pr}^c - \text{Pr}^b|}{\max\{\text{Pr}^c, \text{Pr}^b\}}\le 2 \cdot \left|\mu^{b,c}_{\pi(i)}(\ell)\right|.$$ If $(z, \pi(i), b) \in G^{[u]}(T)$ is uneven, then the weight is at least ${m}/({m+1})$; if $(z, \pi(i), b) \in G^{[\kappa]}(T)$ for some $\kappa\ge 1$, then the weight is at least $m^{-\kappa}$. This finishes the proof of the lemma.
\end{proof}

\subsection{Case 1: Graph with Uneven Edges}\label{sec:haha2}

We assume there are parameters $\zeta', \xi$, and $d \geq 1$ such that with probability at least $\zeta'$ over $T \sim \mathcal{S}(t)$,
\begin{equation} \label{equation35}
\Pr_{x \sim p_{\overline{T}}} \left[d \leq \text{outdeg}(x, G^{[u]}(T)) \leq 2d \right] \geq \xi.
\end{equation}
Notice that $\zeta' = \zeta/(2 \log (nm))$ (see Case 1 of Section \ref{section_bucketing}), so that (\ref{equation32}) implies:
\begin{equation} \label{equation36}
\sqrt{d} \cdot \xi \gtrsim \frac{\alpha}{\zeta'\log (1 / \alpha)\cdot m^{3}\log^4(nm)}.
\end{equation}

We define $(\tau,y)$, where $\tau$ is a $t$-sequence and $y\in \mathbb{Z}_M$, to be $t$-contributing or $(t+1)$-contributing:
\begin{definition} \label{def:t-contributing-uneven-case}
Let $\tau$ be a $t$-sequence of distinct indices from $[n]$ and let $y \in \mathbb{Z}_M$. We say the pair $(\tau, y)$ is $t$\textit{-contributing} if the restricted distribution $p_{| \rho(\tau, y)}$ satisfies
$$
\Big\|\mu\left(p_{|\rho(\tau,y)}\right)\Big\|_2
\geq \frac{1}{m}\cdot\sqrt{\frac{d}{32}},$$
and we say $(\tau, y)$ is $(t+1)$-contributing otherwise. 
\end{definition}

Lemma \ref{Lemma3.1} would follow if there are many $t$-contributing pairs $(\tau,y)$.
The next lemma gives us the tool we need to obtain $t$-contributing pairs:

\begin{lemma} \label{Lemma3.9}
Let $\pi$ be a $(t+1)$-sequence of distinct indices from $[n]$, and $y \in \mathbb{Z}_M$ be in the support of $p$. If there are distinct  $i_1, ..., i_{d+1} \in [t+1]$ such that for each  $k \in [d+1]$ there is a $b_k \in \mathbb{Z}_{m_{\pi(i_k)}}$ such that 
$$\left(y_{\overline{S(\pi_{-{i_k}})}}, \pi(i_k), b_k\right) \in G^{[u]}\left(S(\pi_{-i_k})\right),$$
then $(\pi_{-i_k}, y)$ is a $t$-contributing pair for at least one of the indices $k \in [d+1]$.
\end{lemma}

We begin by proving the following claim.

\begin{claim} \label{Lemma3.9Helper}
Let $\pi$ be a $(t+1)$-sequence of distinct indices from $[n]$, and let $y \in \mathbb{Z}_M$ be in the support of $p$. If there are~$i\ne j\in [d+1]$, $b_i \in \mathbb{Z}_{m_{\pi(i)}}$ and $b_j \in \mathbb{Z}_{m_{\pi(j)}}$ such that 
$$\left(y_{\overline{S(\pi_{-{i}})}}, \pi(i), b_i\right) \in G^{[u]}\left(S(\pi_{-i})\right)\quad\text{and}\quad \left(y_{\overline{S(\pi_{-{j}})}}, \pi(j), b_j\right) \in G^{[u]}\left(S(\pi_{-j})\right),$$
then there must exist either a $c_i \in \mathbb{Z}_{m_{\pi(i)}}$ or $c_j \in \mathbb{Z}_{m_{\pi(j)}}$  such that either 
$$\left|\mu^{c_i,y_{\pi(i)}}_{\pi(i)}\left(p_{| \rho(\pi_{-j}, y)}\right)\right| \geq \frac{1}{4m}\quad\text{or}\quad \left|\mu^{c_j,y_{\pi(j)}}_{\pi(j)}\left(p_{| \rho(\pi_{-i}, y)}\right)\right| \geq \frac{1}{4m}.$$
\end{claim}

\noindent\textit{Intuition behind the proof:} We proceed by contradiction and assume for all $c_i$ and $c_j$, we have
\begin{equation}\label{hehe112}\left|\mu^{c_i,y_{\pi(i)}}_{\pi(i)}\left(p_{| \rho(\pi_{-j}, y)}\right)\right|< \frac{1}{4m}\quad\text{and}\quad \left|\mu^{c_j,y_{\pi(j)}}_{\pi(j)}\left(p_{| \rho(\pi_{-i}, y)}\right)\right| < \frac{1}{4m}.\end{equation}
Then all the probabilities $$\Pr_{x \sim p}\left[x_{\pi(i)} = c_i, x_{\pi(j)} = y_{\pi(j)}, x_{\overline{S(\pi)}} = y_{\overline{S(\pi)}}\right]\quad\text{and}\quad \Pr_{x \sim p}\left[x_{\pi(i)} = y_{\pi(i)}, x_{\pi(j)} = c_j, x_{\overline{S(\pi)}} = y_{\overline{S(\pi)}}\right]$$ must be close to the probability $$\Pr_{x \sim p}\left[x_{\pi(i)} = y_{\pi(i)}, x_{\pi(j)} = y_{\pi(j)}, x_{\overline{S(\pi)}} = y_{\overline{S(\pi)}}\right].$$
In particular, this implies that 
\begin{equation}\label{hehe111}
\Pr_{x \sim p}\left[x_{\pi(i)} = b_i, x_{\pi(j)} = y_{\pi(j)}, x_{\overline{S(\pi)}} = y_{\overline{S(\pi)}}\right]\approx \Pr_{x \sim p}\left[x_{\pi(i)} = y_{\pi(i)}, x_{\pi(j)} = y_{\pi(j)}, x_{\overline{S(\pi)}} = y_{\overline{S(\pi)}}\right].
\end{equation}
However, this leads to a contradiction since we can use
$$
\Pr_{x \sim p}\left[x_{\pi(i)} = y_{\pi(i)}, x_{\pi(j)} = c_j, x_{\overline{S(\pi)}} = y_{\overline{S(\pi)}}\right]\approx \Pr_{x \sim p}\left[x_{\pi(i)} = y_{\pi(i)}, x_{\pi(j)} = y_{\pi(j)}, x_{\overline{S(\pi)}} = y_{\overline{S(\pi)}}\right].
$$
and the unevenness of the edge
  $$\left(y_{\overline{S(\pi_{-{i}})}}, \pi(i), b_i\right)$$
to show that the LHS of (\ref{hehe111}) is smaller than its RHS, a contradiction.

\begin{proof}[Proof of Claim \ref{Lemma3.9Helper}] 
Suppose towards a contradiction that (\ref{hehe112}) holds. 
Also for notational convenience we assume that $\pi(i)=1$ and $\pi(j)=2$. From the first part of (\ref{hehe112}) we have
$$\left|\frac{ \Pr_{x \sim p}[x_1 = y_1, x_2 = y_2, x_{\overline{S(\pi)}} = y_{\overline{S(\pi)}}] - \Pr_{x \sim p}[x_1 = c_i, x_2 = y_2, x_{\overline{S(\pi)}} = y_{\overline{S(\pi)}}]}{\Pr_{x \sim p}[x_1 = y_1, x_2 = y_2, x_{\overline{S(\pi)}} = y_{\overline{S(\pi)}}] + \Pr_{x \sim p}[x_1 = c_i, x_2 = y_2, x_{\overline{S(\pi)}} = y_{\overline{S(\pi)}}]}\right|< \frac{1}{4m}.$$
Letting $\gamma :=\Pr_{x \sim p}[x_1 = y_1, x_2 = y_2, x_{\overline{S(\pi)}} = y_{\overline{S(\pi)}}]>0$ (note that $\gamma>0$ because we assumed that $y$ is in the support of $p$),
this implies for any $c_i \neq y_1$, we have
\begin{equation} \label{BoundPiI}
    \begin{aligned}
    \gamma \cdot \frac{1 - \frac{1}{4m}}{1 + \frac{1}{4m}} < \Pr_{x \sim p}\left[x_1 = c_i, x_2 = y_2, x_{\overline{S(\pi)}} = y_{\overline{S(\pi)}}\right]  
    < \gamma\cdot \frac{1 + \frac{1}{4m}}{1 - \frac{1}{4m}}.
    \end{aligned}
\end{equation}
Similarly, for any $c_j \neq y_{\pi(j)}$, we have
\begin{equation} \label{BoundPiJ}
    \begin{aligned}
 \gamma \cdot \frac{1 - \frac{1}{4m}}{1 + \frac{1}{4m}} < \Pr_{x \sim p}\left[x_1 = y_1, x_2 = c_j, x_{\overline{S(\pi)}} = y_{\overline{S(\pi)}}\right]  
    < \gamma \cdot \frac{1 + \frac{1}{4m}}{1 - \frac{1}{4m}}.
    \end{aligned}
\end{equation}
In particular, setting $c_i=b_i$, we have
\begin{equation}\label{hehe113}
\Pr_{x\sim p}\left[x_1= b_i, x_{\overline{S(\pi)}} = y_{\overline{S(\pi)}}\right]
\ge \Pr_{x\sim p}\left[x_1= b_i, 
x_2=y_2, x_{\overline{S(\pi)}} = y_{\overline{S(\pi)}}\right]
>\gamma\cdot  \frac{1-\frac{1}{4m}}{1+\frac{1}{4m}}.
\end{equation}

On the other hand, given that
$$\left(y_{\overline{S(\pi_{-{i}})}}, \pi(i), b_i\right) \in G^{[u]}\left(S(\pi_{-i})\right),$$ we have
$$\frac{|\Pr_{x\sim p}[x_1 = y_1, x_{\overline{S(\pi)}} = y_{\overline{S(\pi)}}]- \Pr_{x\sim p}[x_1 = b_i, x_{\overline{S(\pi)}} = y_{\overline{S(\pi)}}]|}{\max \{\Pr_{x\sim p}[x_1 = y_1, x_{\overline{S(\pi)}} = y_{\overline{S(\pi)}}], \Pr_{x\sim p}[x_1 = b_i, x_{\overline{S(\pi)}} = y_{\overline{S(\pi)}}]\}} \geq \frac{m}{m+1}.$$ 
By the orientation of uneven edges, we have
$$\Pr_{x\sim p}\left[x_1 = y_1, x_{\overline{S(\pi)}} = y_{\overline{S(\pi)}}\right] \geq \Pr_{x\sim p}\left[x_1 = b_i, x_{\overline{S(\pi)}} = y_{\overline{S(\pi)}}\right].$$ As a result, we have 
\begin{equation} \label{UBonBi}
    \Pr_{x\sim p}\left[x_1 = b_i, x_{\overline{S(\pi)}} = y_{\overline{S(\pi)}}\right] \leq \frac{1}{m+1} \cdot \Pr_{x\sim p}\left[x_1 = y_1, x_{\overline{S(\pi)}} = y_{\overline{S(\pi)}}\right].
\end{equation}

Finally, using (\ref{BoundPiJ}), we have
$$
\Pr_{x\sim p}\left[x_1 = y_1, x_{\overline{S(\pi)}}\right]\le  \gamma\left(1+(m_2-1)\cdot \frac{1+\frac{1}{4m}}{1-\frac{1}{4m}}\right)
\le \gamma \left(1+(m-1)\cdot \frac{1+\frac{1}{4m}}{1-\frac{1}{4m}}\right)
$$
This together with (\ref{hehe113}) and (\ref{UBonBi}) lead to a contradiction.
\end{proof}

We are now ready to prove Lemma \ref{Lemma3.9}.

\begin{proof}[Proof of Lemma \ref{Lemma3.9}]
Suppose there are $d+1$ distinct indices $i_1, ..., i_{d+1} \in [t+1]$ such that for each index $k \in [d+1]$ there exists a $b_k \in \mathbb{Z}_{m_{\pi(i_k)}}$ such that 
$$\left(y_{\overline{S(\pi_{-{i_k}})}}, \pi(i_k), b_k\right) \in G^{[u]}\left(S(\pi_{-i_k})\right).$$ 
By Claim \ref{Lemma3.9Helper},  for each pair $i, j$ of the $d+1$ indices, there exists $c_i$ or $c_j$ such that either 
\begin{equation}\label{hehe114}\left|\mu^{c_i,y_{\pi(i)}}_{\pi(i)}\left(p_{| \rho(\pi_{-j}, y)}\right)\right| \geq \frac{1}{4m}\quad\text{or}\quad \left|\mu^{c_j,y_{\pi(j)}}_{\pi(j)}\left(p_{| \rho(\pi_{-i}, y)}\right)\right| \geq \frac{1}{4m}.\end{equation}

Construct a graph $G$ as follows. Let its vertex set be $i_1,\ldots,i_{d+1}$. For each pair $i, j$, create a directed edge from $i$ to $j$ if 
the second part of (\ref{hehe114}) holds and 
create a directed edge from $j$ to $i$ if 
if the second part of (\ref{hehe114}) holds.
If both inequalities hold, orient the edge between $i$ and $j$ arbitrarily. So $G$ a directed version of the complete undirected graph over $i_1,\ldots,i_{d+1}$.

It is easy to show that $G$ has a vertex with out-degree $\geq d/2$; otherwise the total number of edges is $<(d+1)\cdot (d/2)$, a contradiction. 
To finish the proof, let $k$ be an index among $i_1,\ldots,i_{d+1}$ that has out-degree at least $d/2$.
Then we have 
$$
\Big\|\mu\left(p_{|\rho(\pi_{-k},y)}\right)\Big\|_2^2
\geq \frac{1}{(4m)^2} \cdot \frac{d}{2},$$
and the lemma follows.
\end{proof}

We now use Lemma \ref{Lemma3.9} to prove Case 1 of Lemma \ref{Lemma3.1}. We will need to lowerbound the expectation of $\|\mu(p_{|\rho(\pi,y)|})\|_2$ as $(\pi,y) \sim \mathcal{D}'(t+1, p)$.
For each $i\in [t+1]$,
let $X_i$ be the indicator random variable that is set to $1$ when the following event holds:
\begin{equation} \label{equation40}
\left\{\left(y_{\overline{S(\pi_{-{i}})}}, \pi(i), b\right) \in G^{[u]}\left(S(\pi_{-i})\right) \text{ for some } b \in \mathbb{Z}_{m_{\pi(i)}} \text{\ and\ }(\pi_{-i}, y)\text{ is } (t+1)\text{-contributing}\right\}.
\end{equation}

First, combining Lemma \ref{Lemma3.7} and the first part of the event above gives us the following inequality:
$$
\Big\|\mu\left(p_{|\rho(\pi,y)}\right)\Big\|_2 \geq \frac{m}{2(m+1)}\cdot  \sqrt{X_1 + \dots + X_{t+1}}.
$$
This is because Lemma \ref{Lemma3.7} implies that
\begin{align*}
\Big\|\mu\left(p_{|\rho(\pi,y)}\right)\Big\|_2^2&\ge 
 \sum_{i \in [t+1]} \sum_{c \in \mathbb{Z}_{m_{\pi(i)}}} \left(\frac{m}{2(m+1)}\right)^2 \cdot \mathbbm{1} \left\{ \left(y_{\overline{S(\pi_{-i})}}, \pi(i), c\right) \in G^{[u]}\left(S(\pi_{-i})\right) \right\}\\[0.5ex]
&\ge \left(\frac{m}{2(m+1)}\right)^2\cdot \big(X_1+\cdots+X_{t+1}\big).
\end{align*}

Second, we can use Lemma \ref{Lemma3.9} and the second part of (\ref{equation40}) to see that $X_1 + \dots + X_{t+1}$ is at most $d$ with probability $1$. Therefore, we can obtain the following expression:
$$\mathbb{E}_{(\pi, y)} \left[ \Big\|\mu\left(p_{|\rho(\pi,y)}\right)\Big\|_2  \right] \gtrsim \mathbb{E}_{(\pi, y)} \big[ X_1 + \dots + X_{t+1} \big] \cdot \frac{1}{\sqrt{d}} \cdot \frac{m}{2(m+1)},$$
where $(\pi,y)\sim \mathcal{D}'(t+1,p)$.

What remains is bounding the probability of $X_i = 1$ for each $i\in [t+1]$. The proof of this part is exactly the same as in \cite{canonne2021random}, and we include it for completeness.
We claim that for each $i \in [t + 1]$, 
\begin{equation} \label{equation41}
\Pr_{(\pi, y)}\big[X_i = 1\big] \geq \left(\zeta' \xi - \Pr_{(\pi, y)}\big[(\pi_{-i}, y) \text{ is  } t\text{-contributing} )\big]\right) \cdot \frac{d}{nm}.
\end{equation}

Let us consider drawing $\pi$ and $y$ by drawing $y$ and $\pi_{-i}$ first and then $\pi(i)$. We define event $F$ over $y$ and $\pi_{-i}$ as follows:
$$
\text{Event $F$: $S(\pi_{-i})$ as $T$ and $y_{\overline{S(\pi_{-i}})}$ as $x$ satisfy (\ref{equation35}) and $(\pi_{-i}, y)$ is $(t+1)$-contributing.}
$$
From our assumption at the beginning of Case 1, the first part of $F$ occurs with probability at least $\zeta' \xi$. Therefore, the probability of $F$ is at least
$$\zeta' \xi - \Pr_{(\pi, y)}\big[(\pi_{-i}, y) \text{ is } t\text{-contributing} \big].$$
Conditioning on $\pi_{-i}$ and $y$ satisfying $F$, $\pi(i)$ (together with $\pi_{-i}$ and $y$) leads to $X_i = 1$ if there exists $b$ such that 
$(y_{\overline{S(\pi_{-{i}})}}, \pi(i), b) \in G^{[u]}(S(\pi_{-i}))$. The probability of this is at least ${d}/({n(m-1)})$.

Continuing from (\ref{equation41}), next we can observe that the probability of $(\pi_{-i}, y)$ being $t$-contributing is the same as $(\tau, y)$ being $t$-contributing, where $(\tau,y)\sim \mathcal{D}'(t,p)$.
Putting everything together yields 
\begin{equation} 
\mathbb{E}_{(\pi,y)}\left[\Big\|\mu\left(p_{|\rho(\pi,y)}\right)\Big\|_2 \right] \gtrsim   \frac{1}{\sqrt{d}} \cdot (t+1) \cdot \left(\zeta'\xi - \Pr_{(\tau, y)}\big[(\tau, y) \text{ is } t\text{-contributing}\big] \right) \cdot \frac{d}{n m}.
\label{equation42}
\end{equation}
Thus, either the probability of $(\tau, y)$ being $t$-contributing is at least $\zeta' \xi/2$, in which case we have
$$\mathbb{E}_{(\tau, y) \sim \mathcal{D}'(t, p)} \left[ \Big\|\mu\left(p_{|\rho(\tau,y)}\right)\Big\|_2 \right] 
\gtrsim \zeta'  \xi \cdot \frac{\sqrt{d}}{m} \gtrsim  \frac{\alpha}{m^4\cdot \log^4 (nm) \log (1/\alpha)}$$
using (\ref{equation36}) for the last inequality, from which (\ref{equation17}) follows as $t+1\le n$. Or, (\ref{equation42}) can be lowerbounded by:
$$\frac{1}{\sqrt{d}} \cdot (t+1) \cdot {\zeta' \xi} \cdot \frac{d}{nm} \gtrsim \frac{t}{n} \cdot \frac{\alpha}{m^4\cdot \log^4 (nm) \log (1/\alpha)}.$$
This finishes the proof of Case 1.

\subsection{Case 2: Graph with Even Edges} \label{Case_2_section}

From Case 2 of Section \ref{section_bucketing}, we assume that there are parameters
$\kappa\in  [O(\log(nm/\alpha))]$,
$\zeta', \xi$, and $d$ such that with probability at least $\zeta'$ over the draw of $T \sim \mathcal{S}(t)$, we have:
$$\mathbb{P}_{x \sim p_{\overline{T}}}\left[d \leq \text{outdeg}(x, G^{[\kappa]}(T)) \leq 2d \right] \geq \xi,$$
where  (\ref{equation33}) implies that:
$$\sqrt{d} \cdot \xi \gtrsim \frac{m^{\kappa}\alpha}{\zeta'\log (1/\alpha)\cdot m^{3.5} \cdot \log^2 n \log^2(nm) \log^2(nm / \alpha) }.$$

We introduce a notion of $t$-contributing and $(t+1)$-contributing restrictions: 

\begin{definition}  \label{def:t-contributing-even-case}
Let $\gamma \geq 1$ be a parameter to be fixed later. (We will set $\gamma$ to be $d$ at the end but we keep it as a parameter for now.) 
A restriction $\rho$ with $t$ stars is said to be $t$-contributing if
$$
\Big\|\mu\left(p_{|\rho}\right)\Big\|_2
\geq \frac{\sqrt{\gamma}}{m^{\kappa + 1}},$$
and we say that $\rho$ is $(t+1)$-contributing otherwise.
\end{definition}

\begin{lemma} \label{Lemma3.12}
Let $\pi$ be a $(t+1)$-sequence of distinct indices from $[n]$ and let $y \in \mathbb{Z}_M$ be in the support of $p$. If $i \in [t+1]$ satisfies that $\rho(\pi_{-i},y^{(\pi(i))\to a})$ is $(t+1)$-contributing for all $a\in \mathbb{Z}_{m_{\pi(i)}}$, then
$$\sum_{j \in [t+1]\setminus \{i \}} \sum_{c,d \in \mathbb{Z}_{m_{\pi(j)}}}\left( \mu^{c,d}_{\pi(j)} \left(p_{| \rho(\pi, y)}\right) \right)^2 < \frac{\gamma}{m^{2 \kappa + 2}}.$$
\end{lemma}
\begin{proof}
Let $$\Pr(a) := \Pr_{x \sim p} \big[x_{\pi(i)} = a \hspace{0.05cm}|\hspace{0.05cm} x_{\overline{S(\pi)}} = y_{\overline{S(\pi)}}\big]$$ for each $a \in \mathbb{Z}_{m_{\pi(i)}}$.
We note that for any $j\ne i$ and $c,d\in \mathbb{Z}_{m_{\pi(j)}}$, $$\mu^{c,d}_{\pi(j)}\left(p_{|\rho(\pi, y)}\right) = 
\sum_{a \in \mathbb{Z}_{m_{\pi(i)}}} \Pr(a) \cdot \mu^{c,d}_{\pi(j)}\left(p_{| \rho(\pi_{-i}, y^{(\pi(i)) \to a})}\right).$$
By Jensen's inequality, we have
\begin{equation*} 
   \sum_{j \in [t+1]\setminus \{i \}} \sum_{c,d \in \mathbb{Z}_{m_{\pi(j)}}} \left( \mu^{c,d}_{\pi(j)} \left(p_{| \rho(\pi, y)}\right)  \right)^2 \leq \sum_{j \in [t+1]\setminus \{i \}} \sum_{c,d \in \mathbb{Z}_{m_{\pi(j)}}} \sum_{a \in \mathbb{Z}_{m_{\pi(i)}}} \Pr(a)\cdot  \left( \mu^{c,d}_{\pi(j)}\left(p_{| \rho(\pi_{-i}, y^{(\pi(i)) \to a})}\right) \right)^2
\end{equation*}

Since we assume each pair $(\pi_{-i}, y^{(\pi(i)) \to a})$ is $(t+1)$-contributing, we have
$$\sum_{j \in [t+1] \setminus \{i\}} \sum_{c,d \in \mathbb{Z}_{m_{\pi(j)}}} \left(\mu^{c,d}_{\pi(j)} \left( p_{| \rho(\pi_{-i}, y^{(\pi(i)) \to a})}\right)\right)^2 < \frac{\gamma}{m^{2\kappa + 2}}$$ for each such $\rho(\pi_{-i}, y^{(\pi(i)) \to a})$.
The lemma follows using $\sum_a\Pr(a)=1$.
\end{proof}

Similar to Case 1 we would like to lowerbound the expectation of $\|\mu(p_{|\rho(\pi,y)})\|_2$
 as $(\pi,y) \sim \mathcal{D}'(t + 1, p)$. 
Let us introduce the following indicator random variable $X_i$ for each $i \in [t+1]$. $X_i$ equals $1$ when the following event $F_i$ holds:
\begin{flushleft}\begin{quote}
Event $F_i$ on $(\pi,y)\sim\mathcal{D}'(t+1,p)$: There exists a $b \in \mathbb{Z}_{m_{\pi(i)}}$ such that $$\left(y_{\overline{S(\pi_{-i})}}, \pi(i), b\right) \in G^{[\kappa]}\left(S(\pi_{-i})\right),$$ and for all $c \in \mathbb{Z}_{m_{\pi(i)}}$,  we have that 
$\rho(\pi_{-i}, y^{ \pi(i) \to c })$ is $(t+1)$-contributing. 
\end{quote}\end{flushleft}

Combining Lemma \ref{Lemma3.7} and the first part of the event $F_i$ gives us
$$
\Big\|\mu\left(p_{|\rho(\pi,y)}\right)\Big\|_2
\geq \frac{1}{2 m^\kappa} \cdot \sqrt{X_1 + \dots + X_{t+1}}.$$
Combining this inequality with Lemma \ref{Lemma3.12} and the second part of Event $F_i$ implies that $X_1 + \dots + X_{t+1}$ is at most $\lceil 8\gamma/m^2\rceil$ with probability $1$. This is because if the sum is more than $\lceil 8\gamma/m^2\rceil$, there are two $X_i = X_j = 1$, from which Lemma \ref{Lemma3.12} implies that 
$$  \frac{1}{4m^{2\kappa}} \cdot (X_1 + \dots + X_{t+1}) \leq \Big\|\mu\left(p_{|\rho(\pi,y)}\right)\Big\|_2^2=\sum_{i \in [t+1]} \sum_{c ,d\in \mathbb{Z}_{m_{\pi(i)}}} \left(\mu^{c,d}_{\pi(i)}\left(p_{|\rho (\pi, y)}\right)\right)^2 < 2\cdot  \frac{\gamma}{m^{2 \kappa + 2}},$$
which implies that $X_1 + \dots + X_{t+1} < {8 \gamma}/{m^2}$. 
As a result, we have
$$\mathbb{E}_{(\pi, y)} \left[\Big\|\mu\left(p_{|\rho(\pi,y)}\right)\Big\|_2 \right] \gtrsim \frac{1}{m^{\kappa}} \cdot \mathbb{E}_{(\pi, y)} \big[ X_1 + \dots + X_{t+1} \big] \cdot \sqrt{\frac{1}{\lceil 8\gamma/m^2\rceil}}.$$
What remains is bounding the probability that $X_i = 1$.

To do so, we will need to introduce some notation. Let $T$ be a size-$t$ subset of $[n]$ and let $\smash{z \in \bigtimes_{i \in \overline{T}} \mathbb{Z}_{m_i}}$. We write $\rho(z)$ to denote the restriction $\smash{\rho \in \bigtimes_{i = 1}^n (\mathbb{Z}_{m_i} \cup \{* \})}$ with $\rho_i = z_i$ for all $i \in \overline{T}$ and $\rho_i = \ast$ for all $i \in T$. 
Define the following two disjoint subsets for each size-$t$ subset $T$:
\begin{align*}
    A_T &= \left\{ z\in  \bigtimes_{i \in \overline{T}} \mathbb{Z}_{m_i}: d \leq \text{outdeg}\left(z, G^{[\kappa]}(T)\right) \leq 2d \text{ and } \rho(z) \text{ is } (t+1)\text{-contributing} \right\}\\[0.8ex]
 B_T &= \left\{ w \in \bigtimes_{i \in \overline{T}} \mathbb{Z}_{m_i}: \rho(w) \text{ is } t\text{-contributing} \right\}.
\end{align*}
It is not hard to see that 
  the probability of $X_i = 1$ (i.e. the event $F_i$ on $(\pi,y)\sim \mathcal{D}'(t+1,p)$) is at least the probability of the following event $E$, where we first draw a $t$-subset $T$ of $[n]$ uniformly at random, then $z\sim p_{\overline{T}}$, and finally draw $i$ from $\overline{T}$ uniformly at random:
$$\text{Event } E: \exists\hspace{0.05cm} b\in \mathbb{Z}_{m_i} : \left(z, i, b\right) \in G^{[\kappa]}(T), z \in A_T \text{, and } \forall\hspace{0.05cm} c \in \mathbb{Z}_{m_i}: z^{(i) \to c} \not \in B_T.$$

We now lowerbound the probability of $E$ over $T,z$ and $i$. First, we write the probability as follows:
\begin{align} \label{equation45}
    &\Pr_{T, z, i}\big[E\big] = \Pr_{T, z, i}\left[\exists\hspace{0.05cm} b \in \mathbb{Z}_{m_i} : (z, i, b) \in G^{[\kappa]}(T) \text{\ and\ } z \in A_T\right] \\[0.5ex]
    &\hspace{0.5cm}- \Pr_{T, z, i}\left[\exists\hspace{0.05cm} b \in \mathbb{Z}_{m_i}: (z, i, b) \in G^{[\kappa]}(T), z \in A_T \text{\ and } \exists\hspace{0.05cm} c \in\mathbb{Z}_{m_i} : z^{(i) \to c} \in B_T \right].\nonumber
\end{align}
The first probability on the right hand side of (\ref{equation45}) is at least:
$$\left( \zeta' \xi -  \Pr_{T,z}\big[\rho(z) \text{ is } t\text{-contributing} \big]\right) \cdot \frac{d}{(n-t)m}.$$
To see this, we first draw $T$ and $z$ and then impose the condition that $z \in A_T$. Similarly to the Case 1 arguments, the probability of such an event is at least:
$$\zeta' \xi - \Pr_{T, z}\big[\rho(z)\ \text{is $t$-contributing}\big]. 
$$
We then draw $i$ from $\overline{T}$. The probability of getting an  $(z, i, b)\in G^{[\kappa]}(T)$ for some $b$ is at least $$\frac{d}{(n-t)(m-1)}> \frac{d}{(n-t)m}.$$ 

Next we upperbound the probability that is being subtracted in (\ref{equation45}). It can be written as:
$$\hspace{-0.5cm}\frac{1}{\binom{n}{t}} \sum_{T \in \mathcal{P}(t)} \sum_{z} \Pr_{x \sim p_{\overline{T}}} [x = z]  \sum_{i \in \overline{T}} \frac{1}{|\overline{T}|}\cdot\mathbbm{1} \left\{
\exists\hspace{0.05cm} b \in \mathbb{Z}_{m_i}: (z, i, b) \in G^{[\kappa]}(T), z \in A_T, \exists\hspace{0.05cm} c \in\mathbb{Z}_{m_i} : z^{(i) \to c} \in B_T
\right\},$$
where the sum of $z$ is over $z\in \bigtimes_{i \in \overline{T}} \mathbb{Z}_{m_i}$.
By a union bound, we can write this as:
\begin{align*}
&\hspace{-0.5cm}\frac{1}{\binom{n}{t}(n-t)} \sum_{T \in \mathcal{P}(t)} \sum_{z } \Pr_{x \sim p_{\overline{T}}} [x = z]  \sum_{i \in \overline{T}}  \sum_{b \in \mathbb{Z}_{m_i}} \sum_{c \in \mathbb{Z}_{m_i}}\mathbbm{1}\left\{(z, i, b) \in G^{[\kappa]}(T)  \land z \in A_T \land z^{(i) \to c}  \in B_T \right\}
\\&\hspace{-0.5cm}=
\frac{1}{\binom{n}{t}(n-t)} \sum_{T \in \mathcal{P}(t)}  \sum_{i \in \overline{T}}  \sum_{a \in \mathbb{Z}_{m_i}} \sum_{z: z_i = a } \sum_{b \in \mathbb{Z}_{m_i}} \sum_{c \in \mathbb{Z}_{m_i}} \Pr_{x \sim p_{\overline{T}}} [x = z] \cdot \mathbbm{1}\left\{(z, i, b) \in G^{[\kappa]}(T)  \land z \in A_T \land z^{(i) \to c}  \in B_T \right\}
\end{align*}

We apply a change of variables. Instead of summing over all $a$ and all $\smash{z\in  \bigtimes_{i \in \overline{T}} \mathbb{Z}_{m_i}: z_i = a}$, we sum over all $c \in \mathbb{Z}_{m_i}$, and all $\smash{w \in  \bigtimes_{i \in \overline{T}} \mathbb{Z}_{m_i}: w_i = c}$. Observe that the original variable $z$ changes to $w^{(i) \to a}$, and $z^{(i) \to c}$ becomes $w$. This yields the following expression:
$$\hspace{-1.2cm}\frac{1}{\binom{n}{t}(n-t)} \sum_{T}  \sum_{i \in \overline{T}} \sum_{c \in \mathbb{Z}_{m_i}} \sum_{w : w_i = c} \sum_{a,b \in \mathbb{Z}_{m_i}}   \Pr_{x \sim p_{\overline{T}}} [x = w^{(i) \to a}] \cdot \mathbbm{1}\left\{(w^{(i) \to a}, i, b) \in G^{[\kappa]}(T)  \land w^{(i) \to a} \in A_T \land w  \in B_T \right\}.$$

Observe that if $\smash{(w^{(i) \to a}, w^{(i) \to b} )}$ is not an edge in $G^{[\kappa]}(T)$, then the corresponding indicator variable  above equals zero. Otherwise, if it is an edge in $\smash{G^{[\kappa]}(T)}$, by our construction of $ G^{[\kappa]}(T)$, we must have $(w^{(i) \to a}, w^{(i) \to d}) \not \in G^{[u]}(T)$
for all $d \in \mathbb{Z}_{m_i}$. Therefore, we have 
$$\Pr_{x\sim p_{\overline{T}}}\left[x = w^{(i) \to a}\right] \leq (m+1)\cdot \Pr_{x\sim p_{\overline{T}}}\left[x = w^{(i) \to d}\right]
$$ 
for all $d \in \mathbb{Z}_{m_i}$. Therefore, the expression can be bounded from above by 
\begin{align*}
    &\frac{(m+1)}{\binom{n}{t}(n-t)}  \sum_{T} \sum_{i \in \overline{T}} \sum_{c \in \mathbb{Z}_{m_i}} \sum_{w : w_i = c} \sum_{a ,b\in \mathbb{Z}_{m_i}}   \Pr_{x \sim p_{\overline{T}}} [x = w] \cdot\mathbbm{1}\left\{(w^{(i) \to a}, i, b) \in G^{[\kappa]}(T) \land w^{(i) \to a} \in A_T \land w \in B_T \right\}\\
&= \frac{(m+1)}{\binom{n}{t}(n-t)} \sum_{T} \sum_{w \in B_T}\Pr_{x \sim p_{\overline{T}}} [x = w]  \sum_{i \in \overline{T}}  \sum_{a,b\in \mathbb{Z}_{m_i}}  \mathbbm{1}\left\{(w^{(i) \to a}, i, b) \in G^{[\kappa]}(T) \land w^{(i) \to a} \in A_T \right\}.
\end{align*}

Next, considering the sum of the indicator over all $i \in \overline{T}$ and all $a \in \mathbb{Z}_{m_i}$ yields all possible ways to get to $w^{(i) \to b}$ from $A_T$ in $G^{[\kappa]}(T)$. Thus our expression equals: 
$$\frac{(m+1)}{\binom{n}{t}(n-t)} \sum_{T \in \mathcal{P}(t)} \sum_{w \in B_T} \Pr_{x \sim p_{\overline{T}}} [x = w ] \sum_{b \in \mathbb{Z}_{m_i}} \left[\text{number of edges from } A_T \text{ to } w^{(i) \to b} \text{ in } G^{[\kappa]}(T)\right].$$
Because each vertex in $A_T$ has out-degree at most $2d$ in $G^{[\kappa]}(T)$, we may apply Lemma \ref{Lemma3.4} to conclude that the number of edges from $A_T$ to the string $w^{(i) \to b}$ is at most $2d$, for each $b$. We can therefore say that the probability we subtract is bounded from above by:
$$\frac{2dm(m+1)}{\binom{n}{t}(n-t)}  \sum_{T \in \mathcal{P}(t)} \sum_{w \in B_T}\Pr_{x \sim p_{\overline{T}}} [x = w ] 
= \frac{2d m (m+1)}{n - t} \cdot \Pr_{T, z}\big[z \in B_T\big] .
$$
As a result we have
$$\Pr_{(\pi, y)}\big[ X_i = 1\big]   \geq \Bigg(\zeta' \xi - \left(1 + 4m^3  \right) \cdot \Pr_{T,z} \big[\rho(z) \text{ is } t\text{-contributing} \big] \Bigg) \cdot \frac{d}{(n - t)m}.$$

To conclude the proof of Lemma \ref{Lemma3.1} for Case 2, we set $\gamma = d$. Then we either have
$$\Pr_{T,z}\big[\rho(z) \text{ is } t\text{-contributing}\big] \geq \frac{\zeta' \xi}{ 8m^3},$$
which implies that
$$\mathbb{E}_{\rho\sim \mathcal{D}(t,p)} \left[
\Big\|\mu\left(p_{|\rho}\right)\Big\|_2\right]
\gtrsim\frac{\zeta' \xi}{m^3} \cdot \frac{\sqrt{d}}{m^{\kappa + 1}} 
\gtrsim \frac{\alpha}{m^{7.5}\cdot \log(1/\alpha)\log^2 n \log^2(nm) \log^2(nm/\alpha)}.$$
Or we have
$$\mathbb{E}_{(\pi,y)} \left[
\Big\|\mu\left(p_{|\rho(\pi,y)}\right)\Big\|_2
\right] \gtrsim \frac{t+1}{m^{\kappa}} \cdot 
\frac{1}{\sqrt{d}}\cdot \frac{d\zeta' \xi}{(n-t)m}  \gtrsim \frac{t}{n} \cdot \frac{\alpha}{m^{4.5}\cdot \log(1/\alpha)\log^2 n \log^2(nm) \log^2(nm/\alpha)}.$$
This finishes the proof of Lemma \ref{Lemma3.1}.

\section{Mean Testing over Hypergrids} \label{Mean-test-section}

\begin{algorithm}[t]
\begin{algorithmic}[1]
\Require Dimension $n$, $M=(m_1,\ldots,m_n)$ and sample access to a distribution $p$ over $\mathbb{Z}_M$
\State Draw $N = O(m \log (mn))$ samples $x$ from $p$.
\If {any $i \in [n]$ and $a \in [m_i]$ satisfy ($\#$ samples with $x_i = a$ is $> {2N}/{m_i}$ or $< {N}/({2m_i})$)}
    \State 
    \textbf{return} \textsf{reject}
\Else
    \State \textbf{return} \textsf{accept} 
\EndIf
\end{algorithmic}	
\caption{$\textsc{CoarseTest}(n, M,p)$}\label{algo:coarsetest}
\end{algorithm}

\begin{algorithm}[t]
  \begin{algorithmic}[1]
    \Require Dimension $n$, $M=(m_1,\ldots,m_n)$, $\eps>0$ and sample access to a distribution  $p$ over $\mathbb{Z}_M$
    \State \textsc{CoarseTest}$(n,M,p)$ and return \textsf{reject} if it returns \textsf{reject} 
    \ForAll{$k\in [m^2]$}
        \State Run \textsc{MeanTester}$(n,p^{(k)},\eps/2m)$ for $O(\log m)$ many times
        \State \Return \textsf{reject} if the majority of calls return \textsf{reject}
    \EndFor 
    \State \Return \textsf{accept} 
\end{algorithmic}
  \caption{$\textsc{ProjectedTestMean}(n, M, \eps,p)$}\label{algo:projectedtestmean}
\end{algorithm}

\textsc{ProjectedTestMean} is presented as Algorithm \ref{algo:projectedtestmean}.
It uses a preprocessing subroutine called 
  \textsc{CoarseTest} which is presented as Algorithm \ref{algo:coarsetest}.
It also uses $\textsc{MeanTester}$ from \cite{canonne2021random}. To state the performance guarantee of $\textsc{MeanTester}$,
  we note that the bias vector $\mu(p)$ of a distribution $p$ over $\{-1,1\}^n$ has the following simpler form:  
$$  
\mu_i(p):=\Pr_{x\sim p}\big[x_i=1\big]
-\Pr_{x\sim p}\big[x_i=-1\big].
$$

\begin{theorem}
[\textsc{MeanTester} \cite{canonne2021random}]\label{mainfromprevious}
There is an algorithm (\textsc{MeanTester}) which, given $n$, sample
access to a distribution $p$ over $\{-1,1\}^n$, and a parameter $\eps\in (0,1]$,
draws
$$
O\left(\max\left\{\frac{1}{\eps^2\sqrt{n}},
\frac{1}{\eps}\right\}\right)
$$
many samples from $p$ and 
has the following performance guarantee:
\begin{enumerate}
    \item If $p$ is the uniform distribution, the algorithm outputs accept with probability at least $2/3$; and
\item If $p$ satisfies $\|\mu(p)\|_2\ge \eps$, the algorithm outputs reject with probability at least $2/3$.
\end{enumerate}
\end{theorem}

The main idea behind \textsc{ProjectedTestMean} is to reduce the mean testing of $p$ over $\mathbb{Z}_M$ to that of the following collection of $m^2$ distributions over $\{-1,1\}^n$.

\begin{definition}\label{p_prime_distribution}
Fix an arbitrary ordering of 
  pairs $(c,d)\in \mathbb{Z}_{m_i}^2$ for each $i\in [n]$ (so that we can refer to them as the $k$-th pair, $k=1,\ldots,m_i^2$).
Let $p$ be a distribution over $\mathbb{Z}_M$.
Given any $k\in [m^2]$, we define 
  a distribution $p^{(k)}$ over $\{-1,1\}$ as follows.
For each $i\in [n]$, let $(c_i,d_i)$ be the $\min(k,m_i^2)$-th pair in $\mathbb{Z}_{m_i}^2$.
To draw $z\sim p^{(k)}$, we first draw $x\sim p$ and then set $z_i$ for each $i\in [n]$ to be $1$ if $x_i=c_i$,
$-1$ if $x_i=d_i$, and an independent and uniformly random bit from $\{-1,1\}$ if $x_i\notin \{c_i,d_i\}$.
\end{definition}

We note that sample access to $p^{(k)}$ for any $k$ can be simulated easily, sample by sample, using sample access to $p$. The following simple lemma helps connects $\|\mu(p)\|_2^2$ with $\sum_k \|\mu(p^{(k)})\|_2^2$.

\begin{lemma}\label{lemma666}
Suppose that $p$ satisfies 
\begin{equation}\label{assump}
\frac{1}{4m_i}\le \Pr_{x\sim p}\big[x_i=a\big]\le \frac{4}{m_i}
\end{equation}
for all $i\in [n]$ and $a\in \mathbb{Z}_{m_i}$. Then we have 
$$
\sum_{k\in [m^2]} \left\|\mu(p^{(k)})\right\|_2^2
\ge \frac{1}{4m^2}\cdot \big\|\mu(p)\big\|_2^2.
$$
\end{lemma}
\begin{proof}
Note that every term in $\|\mu(p)\|_2^2$ appears at least once on the LHS, except that it is multiplied by $|\Pr_{x\sim p}[x_i=c]+\Pr_{x\sim p}[x_i=d]|^2$ for some $i\in [n]$ and $c,d\in \mathbb{Z}_{m_i}$.
The latter (without squaring) is at least $1/(2m_i)\ge 1/(2m)$ given the assumption, and the lemma follows.
\end{proof}

The assumption (\ref{assump}) of Lemma \ref{lemma666} can be easily checked by \textsc{CoarseTest} 
(Algorithm \ref{algo:coarsetest}).
The proof of its performance guarantee below is standard using the Chernoff bound.

\begin{lemma} \label{coarsetest-performancelemma}
There is an algorithm (Algorithm \ref{algo:coarsetest}: \textsc{CoarseTest}) which, given $n, M=(m_1,\ldots,m_n)$, and sample access to a distribution $p$ over $\mathbb{Z}_M$, draws $O ( m \log(m n) )$ samples $x \sim p$ and satisfies:
\begin{enumerate}
    \item If $p$ is the uniform distribution, the algorithm outputs $\accept$ with probability at least $1-1/n$.
    \item If there exists an $i \in [n]$ and an $a \in \mathbb{Z}_{m_i}$ such that either$$\Pr_{x \sim p} \left[x_i = a\right] < \frac{1}{4m_i} ~~~ \text{or} ~~~ \mathbb{P}_{x \sim p} \left[x_i = a \right] > \frac{4}{m_i},$$
    then the algorithm outputs \textit{reject} with probability at least $1 - 1/n$.
\end{enumerate}
\end{lemma}

We are now ready to finish the proof of Theorem \ref{projectedtestmean-performancelemma}  on \textsc{ProjectedTestMean}:

\begin{proof}[Proof of Theorem \ref{projectedtestmean-performancelemma}]
The number of samples used by \textsc{ProjectedTestMean} is
$$
O\big(m\log(mn)\big)+O\big(m^2\log m\big)\cdot \max\left(\frac{1}{\sqrt{n}}\cdot \frac{4m^2}{\eps^2},\frac{2m}{\eps}\right).
$$

When $p$ is the uniform distribution, it follows by Lemma \ref{coarsetest-performancelemma} that it is rejected by \textsc{CoarseTest} with probability $o_n(1)$.
Given that $p^{(k)}$ is uniform for every $k\in [m^2]$, by setting the constant hidden in the $O(\log m)$ large enough, it follows from Theorem \ref{mainfromprevious}, Chernoff bound and a union bound over all $k\in [m^2]$ that it is rejected by calls to \textsc{MeanTester} with probability $o_n(1)$.

When $p$ satisfies $\|\mu(p)\|_2\ge \eps m\sqrt{n}$, we consider two cases.
If there exist $i\in [n]$ and $a\in \mathbb{Z}_{m_i}$ such that either $\Pr_x[x_i=a]< 1/(4m_i)$ or $\Pr_x[x_i=a]> 4/m_i$, then $p$ is rejected by \textsc{CoarseTest} with probability at least $1-o_n(1)$.
On the other hand, if this is not the case, then by Lemma \ref{lemma666}, there exists a $k\in [m^2]$ such that $\|\mu(p^{(k)})\|_2\ge \eps\sqrt{n}/(2m)$.
It follows from Theorem \ref{mainfromprevious} and Chernoff bound that with probability at least $1-o_n(1)$, the majority of calls to \textsc{MeanTest} reject $p^{(k)}$. 
\end{proof}

\section{Discussion and Open Problems} \label{section:openproblems}

In this paper, we study uniformity testing over extended high-dimensional domains $[m_1] \times \dots \times [m_n]$ under the subcube conditional query model. In doing so, we prove a robust version Pisier's inequality over hypergrids, which is a result of independent interest. We give an algorithm which makes $\tilde{O}(\text{poly}(m)\sqrt{n}/\epsilon^2)$ queries to a subcube conditional sampling oracle, where $m = \max_i m_i$. This algorithm has nearly optimal sample complexity when $m$ is a constant. The algorithm is a modification of the algorithm of \cite{canonne2021random}, where additional steps are needed in our setting to properly draw the connection to a subroutine performing mean testing over the hypercube. 

We now highlight several compelling open problems related to distribution testing over extended high-dimensional domains $[m_1] \times \dots \times [m_n]$.
\\\\
\textbf{Lower bounds:} There is currently a lack of techniques for studying lower-bounds in the subcube conditional query model setting. To the best of our knowledge, all the known lower-bounds in the subcube conditional query model setting are transferred over from lower-bounds in the standard sampling setting. For example, the lower-bound for uniformity testing in the hypercube setting from \cite{DBLP:journals/corr/CanonneDKS16, DBLP:journals/corr/DaskalakisDK16, bc18} is a consequence of lower-bounds for testing uniformity of product distributions. This lower-bound, which matches the upper-bound given in \cite{canonne2021random} up to poly-logarithmic factors, utilizes the fact that subcube conditional queries do not provide stronger access to product distributions than standard samples do. In the hypergrid setting considered in this paper, the best known lower-bound is also carried over from the standard sampling setting.
\\\\
\textbf{Dependence on $m$:} It remains an interesting question to pin down the dependency on $m$ in the query complexity of uniformity testing of distributions over hypergrids. In this paper, we did not optimize the dependency on $m$. We imagine that, with some work, the exponent in the dependency on $m$ could be brought down to around half of its current value (for example, from $m^{21}$ to $m^{10}$, perhaps). It remains an interesting open question to obtain tight dependence on $m$, and a more challenging open question to obtain tight bounds on all three parameters $n$, $m$ and $\epsilon$. Although  $m^{21}$ may not be optimal, the polynomial dependence on $m$ is a meaningful step in analyzing distribution testing over extended domains. Our analysis demonstrates how to extend and modify techniques from the hypercube domain, reveals new technical challenges, and develops new technical lemmas like the extended Pisier's inequality suitable for the hypergrid domain.
\\\\
\textbf{Identity testing in high dimensions with subcube conditional queries:} There is no direct reduction from identity testing of product distributions or general distributions to uniformity testing. While identity testing in high dimensions has been explored under weaker oracle assumptions (\cite{BCSV}), query complexity bounds for identity testing in the subcube conditional setting are unknown. This is true even for distributions over hypercubes.

\bibliographystyle{alpha}
\bibliography{main}

\appendix

\section{Additional Proofs}\label{appendix-additionalproofs}

\begin{proof}[Proof of Lemma \ref{Lemma1.4}]: Fix any subset $S\subseteq [n]$ of size $t$. Given $u\in 
\bigtimes_{i \in \overline{S}} \mathbb{Z}_{m_i}$, we write $p_{|\rho({S}, u)}$ to denote the distribution supported on $\smash{\bigtimes_{i \in {S}} \mathbb{Z}_{m_i}}$ given by drawing $x \sim p$ conditioned on $x_{\overline{{S}}} = u$. 

We expand the definition of total variation distance to obtain the following expressions:
\begin{align*}
    &2 d_{TV}(p, \mathcal{U})\\ &= \sum_{x \in \mathbb{Z}_M} \left|p(x) - \prod_{i \in [n]} \frac{1}{m_i}\right|\\
&= \sum_{u \in \bigtimes_{i \in \overline{S}} \mathbb{Z}_{m_i}} \sum_{v \in \bigtimes_{i \in S} \mathbb{Z}_{m_i}} \abs{\Pr_{x \sim p} \big[x_{\overline{{S}}} = u \land x_{{S}} = v \big] - \prod_{i \in [n]} \frac{1}{m_i}}\\
&= \sum_{u \in \bigtimes_{i \in \overline{{S}}} \mathbb{Z}_{m_i}} \sum_{v \in \bigtimes_{i \in {S}} \mathbb{Z}_{m_i}} \abs{p_{\overline{S}}(u) \cdot \Pr_{x \sim p} \big[x_{{S}} = v |  x_{\overline{S}} = u\big] - \prod_{i \in [n]} \frac{1}{m_i} }\\
&\leq \sum_{u \in \bigtimes_{i \in \overline{{S}}} \mathbb{Z}_{m_i}} \sum_{v \in \bigtimes_{i \in {S}} \mathbb{Z}_{m_i}}  \abs{p_{\overline{{S}}}(u) \cdot \Pr_{x \sim p}\big[x_{{S}} = v | x_{\overline{{S}}} = u\big] - p_{\overline{{S}}}(u) \cdot \prod_{i \in S} \frac{1}{m_i} } \\&
\hspace{4cm}+ \abs{ p_{\overline{{S}}}(u) \cdot \prod_{i \in S} \frac{1}{m_i} -  \left(\prod_{j \in \overline{{S}}} \frac{1}{m_j}\right) \cdot \left(\prod_{i \in S} \frac{1}{m_i}\right)}\\
&= \sum_{u \in \bigtimes_{i \in \overline{{S}}} \mathbb{Z}_{m_i}} p_{\overline{{S}}}(u) \cdot 2 d_{TV}(p_{|\rho(S, u)}, \mathcal{U}) + \sum_{v \in \bigtimes_{i \in {S}} \mathbb{Z}_{m_i}} \left(\prod_{i \in S} \frac{1}{m_i} \right) \cdot 2 d_{TV}(p_{\overline{{S}}}, \mathcal{U})\\
&= \sum_{u \in \bigtimes_{i\in \overline{{S}}} \mathbb{Z}_{m_i}} p_{\overline{{S}}}(u) \cdot 2 d_{TV}(p_{|\rho({S}, u)}, \mathcal{U}) + 2 d_{TV}(p_{\overline{{S}}}, \mathcal{U}).
\end{align*}
Take the expectation of this inequality over the choice of ${S} \sim {S}_{\sigma}$. We obtain the lemma.
\end{proof}

\begin{proof}[Proof of Theorem \ref{Theorem1.5} assuming Lemma \ref{Lemma3.1}]
The proof is the same  as  \cite{canonne2021random} except for a minor~change.
If $\sigma n \not \in [5, n-5]$, then Theorem \ref{Theorem1.5} is trivially satisfied. Consider $5 \leq \sigma n \leq n - 5$ and let $\delta = \min(\sigma, 1 - \sigma)/2 > 0$.
For notational simplicity, for each $t \in [n-1]$, we write:
$$\alpha_t := \mathbb{E}_{T \sim \mathcal{S}(t)} \big[d_{TV} (p_{\overline{T}}, \mathcal{U}) \big].$$
Note that we can sample $\rho \sim \mathcal{D}_\sigma (p)$ by drawing $k \sim \text{Bin}(n, \sigma)$ and $\rho \sim \mathcal{D}(k, p)$. Let $\beta_t$ be the probability that $k = t$, for $k \sim \text{Bin}(n, \sigma)$. Let
$$
B := \big\{t \in [n-1] : t/n \in [\sigma - \delta, \sigma + \delta] \big\}.
$$
Using a Chernoff bound, we see that $\sum_{t \in B} \beta_t \geq 1 - 2 e^{- \delta n / 5}$. Therefore:
\begin{equation} \label{equation18}
\sum_{t \in B} \beta_t \cdot \alpha_t \geq \mathbb{E}_{T \sim \mathcal{S}_\sigma} \big[ d_{TV}(p_{\overline{T}}, \mathcal{U}) \big] - 2 e^{- \delta n / 5}.
\end{equation}
We then can derive the following series of inequalities:
\begin{align}
\nonumber 2 \cdot &\mathbb{E}_{\rho \sim \mathcal{D}_\sigma(p)} \left[ \left\|\mu(p_{|\rho})\right\|_2\right] \\[1ex]
\nonumber &\geq \sum_{t \in B} \Big( \beta_t \cdot \mathbb{E}_{\rho \sim \mathcal{D}(t, p)} \left[\left\|\mu(p_{|\rho})\right\|_2 \right] + \beta_{t + 1} \cdot \mathbb{E}_{\rho \sim \mathcal{D}(t, p)} \left[ \left\|\mu(p_{|\rho})\right\|_2 \right] \Big)\\[0.5ex] \label{equation19}
&\gtrsim \sum_{t \in B} \beta_t \Big( \mathbb{E}_{\rho \sim \mathcal{D}(t, p)} \left[ \left\|\mu(p_{|\rho})\right\|_2 \right] + \mathbb{E}_{\rho \sim \mathcal{D}(t, p)} \left[ \left\|\mu(p_{|\rho})\right\|_2 \right]\Big)
\\[0.5ex] \label{equation20}
&\gtrsim \frac{\sigma}{m^{7.5} \log^4(nm)} \cdot \sum_{t \in B} \beta_t \cdot \frac{\alpha_t}{\log^2(nm/\alpha_t) \log(1/\alpha_t)}
\\[1ex]
\label{equation21}
&\gtrsim \frac{\sigma}{m^{7.5} \text{poly}(\log (nm))} \cdot \widetilde{\Omega}\left(\mathbb{E}_{S \sim \mathcal{S}_\sigma}\left[d_{TV}(p_{\overline{S}}, \mathcal{U}) \right] - 2 e^{- \min (\sigma, 1 - \sigma) n / 10} \right).
\end{align}
To obtain (\ref{equation19}), we used $t/n \in [\sigma - \delta, \sigma + \delta]$, $\delta = \min(\sigma, 1 - \sigma)/2$ and $\sigma \geq 5/n$. This gives us:
$$\frac{\beta_{t+1}}{\beta_t} = \frac{n-t}{t + 1} \cdot \frac{\sigma}{1 - \sigma} \geq \frac{(1 - \sigma)/2}{(3 \sigma / 2) + (1/n)} \cdot \frac{\sigma}{1 - \sigma} \gtrsim 1.$$
To obtain  (\ref{equation20}) we apply Lemma \ref{Lemma3.1} on each $t \in B$. For (\ref{equation21}) we apply Jensen's inequality (as the function $f(a) = a/(\log^2(nm/a)\log(1/a))$ when $a \neq 0$ and $f(0) = 0$ is convex in $[0, 1]$) and use (\ref{equation18}). 

We can now conclude that Theorem \ref{Theorem1.5} holds, assuming Lemma \ref{Lemma3.1}.
\end{proof}

\begin{proof}[Proof of Theorem \ref{thm:subcube-alg-intro}]
\textit{Proof of (i) (completeness):} 
For the completeness proof, we prove by induction on $n$ that, when $p$ is uniform, \textsc{SubCondUni}$(n, M, \epsilon,p)$ returns $\accept$ with probability at least $2/3$. For the base case when $n = 1$, since (\ref{equation11}) is violated, we just run an algorithm (Lemma 4.20) from \cite{BCSV}, and the completeness of the base case comes from the completeness of this algorithm.

Inductively, assume that the statement holds for dimensions $1$ through $n - 1$. If (\ref{equation11}) is violated,~then the analysis is trivial. For the case when  (\ref{equation11}) is satisfied, we note that the restriction $p_{|\rho}$
is uniform for any $\rho$. Since the total number of restrictions $\rho$ drawn in line \ref{scalgo:sample-restriction} is $\smash{O(L \log L) = \tilde{O}(m^{8.5} \sqrt{n}/ \epsilon^2)}$, 
we may set the constant hidden in the choice of $r$ in line \ref{scalgo:run-ptmean} to be sufficiently large so that line \ref{scalgo:reject-ptmean-majority} returns $\reject$ with probability no larger than $1/6$. Using the inductive hypothesis, we can also say that \textsc{SubCondUni} rejects in line \ref{scalgo:recursive-call-reject} with probability no larger than $1/6$. 
The induction step follows from a union bound.\medskip

\noindent\textit{Proof of (ii) (soundness):}  
Assume that $d_{TV}(p, \mathcal{U}) \geq \epsilon$. We prove by induction on $n$ that \textsc{SubCondUni} rejects with probability at least $2/3$. 
For the general case of the induction step, we know that
either the first case of 
(\ref{equation13}) holds and thus, 
(\ref{equation14}) holds, or the second case of (\ref{equation13}) holds. 

For the first case when (\ref{equation14}) holds, we recall the choice of $L$ and notice that the LHS of (\ref{equation14}) is the expectation of a random variable with values in $[0, 1]$ while the RHS is $1/L$. Using a bucketing argument, we find that there exists a $j \in [ \lceil \log (2L)\rceil]$ such that 
$$\Pr_{\rho \sim \mathcal{D}_\sigma (p)} \left [\frac{\left\|\mu(p_{|\rho}) \right\|_2}{m\sqrt{n}} \geq \frac{1}{2^j} \right]\geq \frac{2^{j-1}}{L \lceil \log (2L) \rceil} \geq \frac{2^j}{4 L \log (2 L)}.$$
So, for one of the restrictions $\rho \sim \mathcal{D}_{\sigma}(p)$ that we sample, the condition in the event above holds, with probability at least $1 - e^{-2} > 5/6$. When this holds, each of the calls to \textsc{ProjectedTestMean} on line \ref{scalgo:run-ptmean} rejects with probability at least $2/3$. So on line \ref{scalgo:reject-ptmean-majority}, \textsc{SubCondUni} rejects with probability at least $2/3$.

For the second case, using bucketing again, there exists a $j \in [ \lceil \log(4/\epsilon)\rceil]$ such that $$\text{Pr}_{\rho \sim \mathcal{D}_\sigma (p)} \left[ d_{TV}(p_{| \rho}, \mathcal{U}) \geq 2^{-j} \right] \geq \frac{\epsilon 2^j}{4 \lceil \log(4/\epsilon)\rceil} \geq\frac{\epsilon 2^j}{8 \log(4/\epsilon)}.$$
Using (\ref{equation11}) and   Chernoff bound, the probability of $|\text{stars}(\rho)| > 2 \sigma n$ is at most $e^{- \sigma n/3} < (\epsilon / 8)^3$. Therefore,
$$\text{Pr}_{\rho \sim \mathcal{D}_\sigma (p)} \left[ d_{TV}(p_{| \rho}, \mathcal{U}) \geq 2^{-j} \text{ and } 0 < |\text{stars}(\rho)| \leq 2 \sigma n \right] \geq \frac{\epsilon2^j}{16 \log(4/\epsilon)}.$$
Since we set $s_j' = (32/\epsilon) \log(4/\epsilon) \cdot2^{-j}$, the probability that at least one restriction $\rho$ satisfies the condition above is at least $5/6$. The probability that majority of the calls to \textsc{ProjectedTestMean} reject this $\rho$ is also at least $5/6$. Therefore,  \textsc{SubCondUni} rejects with probability at least $2/3$.\medskip

\noindent\textit{Query complexity.}
Let $\Phi(n, m,\epsilon)$ denote its query complexity. Using induction on $n$, we will show that
\begin{equation} \label{equation16}
    \Phi(n, \epsilon) \leq C \cdot \frac{m^{21} \sqrt{n}}{\epsilon^2} \cdot \log^c \left(\frac{nm}{\epsilon} \right)
\end{equation}
for some absolute constants $C, c > 0$. Pick $C_1$ and $c_1$ to be two constants such that upon running the algorithm from Lemma 4.20 of \cite{BCSV} on $n, m, \epsilon$ that violate  (\ref{equation11}), the query complexity is at most
$$C_1 \cdot \frac{\sqrt{m}}{\epsilon^2} \cdot \log^{c_1} \left( \frac{1}{\epsilon}\right).$$
Let $C_2$ and $c_2$ be  constants such that the  complexity of the non-recursive componen of \textsc{SubCondUni} (line \ref{scalgo:first-for-start} to line \ref{scalgo:first-for-end}) is bounded by:
$$C_2 \cdot \frac{m^{21} \sqrt{n}}{\epsilon^2} \cdot \log^{c_2} \left(\frac{nm}{\epsilon} \right).$$
Recall that $L = \tilde{O}({m^{8.5} \sqrt{n}}/\eps)$. The expression above follows from the following calculation: 
$$\sum_{j=1}^{\lceil \log(2 L) \rceil} \frac{L \log L}{2^j} \cdot \left(O\big(m^{4}\log (mn)\big)\cdot 
\max\left\{\frac{2^{2j}}{\sqrt{n}},2^j\right\}.\right) \cdot \log\left( \frac{nm}{\epsilon}\right) = \tilde{O}\left(\frac{m^{21} \sqrt{n}}{\epsilon} \right).$$

Finally, set $C := 2 \max(C_1, C_2)$ and $c := \max (c_1, c_2)$.
We are now ready to prove (\ref{equation16}). The base case of $n = 1$ is trivial. In the inductive step the case when (\ref{equation11}) is violated is also trivial. For the general case, with our chosen $C$ and $c$, we have the following bound:
$$\Phi(n, \epsilon) \leq \frac{C}{2} \cdot \frac{m^{21} \sqrt{n}}{\epsilon^2} \cdot \log^c\left(\frac{nm}{\epsilon} \right) + \sum_{j = 1}^{\lceil \log(4/\epsilon) \rceil} s_j' \cdot 100 \log\left(\frac{16}{\epsilon} \right) \cdot \Phi(2 \sigma n,2^{-j}).$$
By the inductive hypothesis and the choice of $\sigma$, we can write each term in the second sum as:
$$C \cdot \frac{32}{\epsilon} \cdot \log \left(\frac{4}{\epsilon} \right) \cdot 100  \log \left(\frac{16}{\epsilon} \right) \cdot m^{21}\sqrt{2 \sigma n}  \cdot 2^j \cdot \log^{c} \left( \sigma n m 2^{j+1} \right) \leq \frac{C}{32} \cdot \frac{m^{21}\sqrt{n}}{\epsilon} \cdot 2^j \cdot \log^c \left(\frac{n m}{\epsilon} \right)$$
using $C_0 \geq (32^2 \cdot 100)^2 \cdot 2$ (from the choice of $\sigma$ in (\ref{eq:defc0})). Lastly, use the following inequality:
$$\sum_{j = 1}^{\lceil \log(4/\epsilon) \rceil} 2^j < 2^{\lceil \log(4/\epsilon) \rceil + 1} \leq \frac{16}{\epsilon}.$$
By induction, we have now proven the query complexity.
\end{proof}

\end{document}